\documentclass{article}

\usepackage[a4paper]{geometry}
\usepackage[utf8]{inputenc}
\usepackage[T1]{fontenc}
\usepackage{lmodern}
\usepackage{amsmath,amssymb,amsthm} \allowdisplaybreaks[1]
\usepackage{mathrsfs}
\usepackage{microtype}
\usepackage{latexsym}
\usepackage[colorlinks]{hyperref}
\usepackage{cleveref}
\usepackage{graphicx}
\usepackage{bookmark}
\usepackage{threeparttable}
\usepackage{booktabs}
\usepackage{multirow}
\usepackage{float}
\usepackage{color}
\usepackage{cite}
\usepackage{bm}

\usepackage{lineno}

\newtheorem{theorem}{Theorem}
\newtheorem{lemma}{Lemma}
\newtheorem{proposition}{Proposition}

\theoremstyle{definition}
\newtheorem{definition}{Definition}

\newtheorem{remark}{Remark}
\newtheorem{example}{Example}

\newcommand{\tabincell}[2]{\begin{tabular}{@{}#1@{}}#2\end{tabular}}

\renewcommand{\vec}[1]{\boldsymbol{#1}} 

\title{Gray Images of Cyclic Codes over $\mathbb{Z}_{p^2}$ and $\mathbb{Z}_p\mathbb{Z}_{p^2}$\thanks{This research is supported by the National Natural Science Foundation of China (12071001).}}
\date{ }
\author{
\small{ Xuan Wang, Minjia Shi}\\ 
\and \small{Key Laboratory of Intelligent Computing and Signal Processing, Ministry of Education, }\\
\small{School of Mathematical Sciences, Anhui University, Hefei, 230601, China}\\
}

\begin{document}

\maketitle


\begin{abstract}
	In the paper, we firstly study the algebraic structures of $\mathbb{Z}_p \mathbb{Z}_{p^k}$-additive cyclic codes and give the generator polynomials and the minimal spanning set of these codes.
    Secondly, a necessary and sufficient condition for the $\mathbb{Z}_p\mathbb{Z}_{p^2}$-additive cyclic code whose Gray image is linear (not necessarily cyclic) over $\mathbb{Z}_p$
    is given.
    Moreover, as for some special families of cyclic codes over $\mathbb{Z}_{p^2}$ and $\mathbb{Z}_p \mathbb{Z}_{p^2}$, the linearity of the Gray images is determined.
\end{abstract}


\section{Introduction} \label{sec:introduction}

Denote by $\mathbb{Z}_p$ and $\mathbb{Z}_{p^k}$ be the rings of integers modulo $p$ and $p^k$, respectively.
Let $\mathbb{Z}_p^n$ and $\mathbb{Z}_{p^k}^n$ be the space of $n$-tuples over $\mathbb{Z}_p$ and $\mathbb{Z}_{p^k}$, respectively.
A $p$-ary code is a non-empty subset $C$ of $\mathbb{Z}_p^n$.
If the subset is a vector space, then we say $C$ is a linear code.
Similarly, a non-empty subset $\mathcal{C}$ of $\mathbb{Z}_{p^k}^n$ is a linear code if $\mathcal{C}$ is a submodule of $\mathbb{Z}_{p^k}^n$.

In $1994$, Hammons \textit{et al.} \cite{1994} showed that some well-known codes can be seen as Gray images of linear codes over $\mathbb{Z}_4$. Later, Calderbank \textit{et al.} \cite{p-adic_cyclic_codes} gave the structure of cyclic codes over $\mathbb{Z}_{p^k}$.
Besides, Ling \textit{et al.}  \cite{LS_Gray_map} studied $\mathbb{Z}_{p^k}$-linear codes, and characterized the linear cyclic codes over $\mathbb{Z}_{p^2}$ whose Gray images are linear cyclic codes.

In $1973$, Delsarte \cite{association_schemes} first defined and studied the additive codes in terms of association schemes.
Borges \textit{et al.} \cite{binary_images_of_z2z4_cyclic,z2z4_generators} studied the generator polynomials, dual codes
and binary images of the $\mathbb{Z}_2 \mathbb{Z}_4$-additive codes.
Since then, a lot of work has been devoted to characterizing $\mathbb{Z}_2\mathbb{Z}_{4}$-additive codes.
Dougherty \textit{et al.} \cite{1-weight-Z2Z4} constructed one
weight $\mathbb{Z}_2\mathbb{Z}_{4}$-additive codes and analyzed their parameters.
Benbelkacem \textit{et al}. \cite{Z2Z4-ACD-LCD} studied
$\mathbb{Z}_2\mathbb{Z}_{4}$-additive complementary dual codes and their Gray images.
In fact, these codes can be viewed as a generalization of the linear complementary dual (LCD for short) codes \cite{LCD} over finite fields.
Bernal \textit{et al.} \cite{Z2Z4-decoding} introduced a decoding method of the $\mathbb{Z}_2\mathbb{Z}_{4}$-linear codes.
More structure properties of $\mathbb{Z}_2\mathbb{Z}_{4}$-additive codes can be found in \cite{Z2Z4-MDS,Z2Z4-kernel-and-rank}.

Moreover, the additive codes over different mixed alphabets have also been intensely studied, for example, $\mathbb{Z}_2\mathbb{Z}_{2}[u]$-additive codes \cite{Z2Z2u}, $\mathbb{Z}_2 \mathbb{Z}_{2^s}$-additive codes \cite{Z2Z2s}, $\mathbb{Z}_{p^r} \mathbb{Z}_{p^s}$-additive codes \cite{ZprZps}, $\mathbb{Z}_2\mathbb{Z}_{2}[u, v]$-additive codes \cite{Z2Z2uv}, $\mathbb{Z}_p \mathbb{Z}_{p^k}$-additive codes \cite{zpzpk_additive_and_dual} and $\mathbb{Z}_p (\mathbb{Z}_p + u \mathbb{Z}_p)$-additive codes \cite{WRS}, and
so on.
It is worth mentioning that $\mathbb{Z}_2\mathbb{Z}_{4}$-additive cyclic codes form an important family of $\mathbb{Z}_2\mathbb{Z}_{4}$-additive codes, many optimal binary codes can be obtained from the images of this family of codes.
In $2014$, Abualrub \textit{et al.} \cite{z2z4_additive_cyclic_codes} discussed $\mathbb{Z}_2\mathbb{Z}_{4}$-additive cyclic codes.
In addition, in \cite{z2z4-linear}, the authors discussed the rank and kernel dimension of the $\mathbb{Z}_2 \mathbb{Z}_4$-linear codes.
Later, Shi \textit{et al.} \cite{wsk_z3z9_linear_rank_kernel} generated the results about rank and kernel to $\mathbb{Z}_3 \mathbb{Z}_9$.
In $2021$, Shi \textit{et al.} \cite{z2z4-quasi-cyclic} defined the $\mathbb{Z}_2 \mathbb{Z}_4$-additive quasi-cyclic codes and gave the conditions for the code to be self-dual and additive complementary dual (ACD for short), respectively.
More details of $\mathbb{Z}_2\mathbb{Z}_{4}$-additive cyclic codes can be found in \cite{z2z4_additive_cyclic_codes,z2z4_generators,binary_images_of_z2z4_cyclic}.

Motivated by \cite{binary_images_of_z2z4_cyclic}, \cite{LS_Gray_map} and \cite{binary_images_of_z4_cyclic}, this paper is devoted to studying the algebraic structures and the Gray images of cyclic codes over $\mathbb{Z}_{p^2}$ and $\mathbb{Z}_p \mathbb{Z}_{p^2}$, respectively.
The paper is organized as follows.
In Section \ref{sec:preliminaries}, we recall the necessary concepts and properties on $\mathbb{Z}_p\mathbb{Z}_{p^k}$-additive cyclic codes and Gray maps.
In section \ref{sec: structures}, we give the algebra structures of a $\mathbb{Z}_p\mathbb{Z}_{p^2}$-additive cyclic code.
In Section \ref{sec:Gray image of cyclic codes}, we study the Gray image of a $\mathbb{Z}_p\mathbb{Z}_{p^2}$-additive cyclic code,
and give a necessary and sufficient condition for the image to be linear.
In Section \ref{sec: examples}, we presented some special families of codes whose Gray images are linear.

\section{Preliminaries}\label{sec:preliminaries}
In the following sections, assume $p$ an odd prime number.
\subsection{Gray Map and Gray Image}\label{subsec:Gray map}

In \cite{LS_Gray_map}, the classical Gray map $\phi: \mathbb{Z}_{p^2} \rightarrow \mathbb{Z}_p^p$ is defined as
\[ \phi(\theta) = \theta_0 (1, 1, \cdots, 1) + \theta_1(0, 1, \cdots, p-1), \]
where $\theta = \theta_0 p + \theta_1 \in \mathbb{Z}_{p^2}$, $0 \leqslant \theta_0, \theta_1 < p$ and $\phi(\theta)$ is a vector of length $p$.
The homogeneous weight \cite{LS_Gray_map} $wt_{hom}$ on $\mathbb{Z}_{p^2}$ is
\[
wt_{hom}(x) =
\begin{cases}
	p, & \text{if} \quad x \in p\mathbb{Z}_{p^2} \setminus \{ 0 \}, \\
	p-1, & \text{if} \quad x \notin p\mathbb{Z}_{p^2}, \\
	0, & \text{if} \quad x = 0.
\end{cases}
\]
The homogeneous distance is $d_{hom}(\vec{u}, \vec{v}) = wt_{hom}(\vec{u} - \vec{v})$, where $\vec{u}, \vec{v} \in \mathbb{Z}_{p^2}^n$.
Then $\phi$ is an isometry from $(\mathbb{Z}_{p^2}^n, d_{hom})$ to $(\mathbb{Z}_{p}^{pn}, d_H)$, where $d_H$ is the Hamming distance on $\mathbb{Z}_{p}^p$.

For positive integers $\alpha$ and $\beta$, let $n = \alpha + p\beta$ and $\Phi: \mathbb{Z}_p^{\alpha} \times \mathbb{Z}_{p^2}^{\beta} \rightarrow \mathbb{Z}_p^n$ be an extension of the Gray map $\phi$,
which is defined as: \[ \Phi(\vec{x}, \vec{y}) = (\vec{x}, \phi(y_1), \phi(y_2), \cdots, \phi(y_{\beta})) \]
for any $\vec{x} \in \mathbb{Z}_p^{\alpha}$, and $\vec{y} = (y_1, y_2, \cdots, y_{\beta}) \in \mathbb{Z}_{p^2}^{\beta}$.
The $p$-ary code $C = \Phi(\mathcal{C})$ is called a $\mathbb{Z}_p \mathbb{Z}_{p^2}$-\textit{linear code}
or the \textit{$p$-ary image} of a $\mathbb{Z}_p\mathbb{Z}_{p^2}$-additive code.
We denote $\langle C \rangle$ the \textit{linear span} of the codewords of $C$.

\subsection{\texorpdfstring{$\mathbb{Z}_p\mathbb{Z}_{p^k}$}{}-Additive Cyclic Codes}\label{subsec:ZpZpk additive cyclic codes}

Recall some definitions about $\mathbb{Z}_p\mathbb{Z}_{p^k}$-additive cyclic codes in \cite{zpzpk_additive_and_dual, TY_Zhu}.

The code $\mathcal{C}$ is called a \textit{$\mathbb{Z}_p\mathbb{Z}_{p^k}$-additive code} with \textit{parameter $(\alpha, \beta)$} if it is a subgroup of
$\mathbb{Z}_p^{\alpha} \times \mathbb{Z}_{p^k}^{\beta}$, where $k\geqslant 2$ is an integer and $\gcd(\beta, p) = 1$.
Obviously, $\mathcal{C}$ is a $\mathbb{Z}_{p^k}$-submodule of $\mathbb{Z}_p^{\alpha} \times \mathbb{Z}_{p^k}^{\beta}$.
Moreover, $\mathcal{C}$ is a subgroup of $\mathbb{Z}_p^{\alpha} \times \mathbb{Z}_{p^k}^{\beta}$, and from the theorem of finite abelian groups, it is isomorphic to an abelian
structure $\mathbb{Z}_p^{\gamma_0+\gamma_k} \times \mathbb{Z}_{p^k}^{\gamma_1} \times \mathbb{Z}_{p^{k-1}}^{\gamma_2} \times \cdots \times \mathbb{Z}_{p^2}^{\gamma_{k-1}}$
where $\gamma_i \geqslant 0$ are integers for $i=0, 1, \cdots, k$.
Thus, the size of $\mathcal{C}$ is $p^{\gamma_0+\gamma_k + k\gamma_1 + (k-1)\gamma_2 +\cdots + 2\gamma_{k-1}}$.

Let $X$ be the set of $\mathbb{Z}_p$ coordinate positions, and $Y$ be the set of $\mathbb{Z}_{p^k}$ coordinate positions, so $|X|=\alpha$ and $|Y|=\beta$.
In general, say the first $\alpha$ positions corresponds to the set $X$ and the last $\beta$ positions corresponds to the set $Y$.
Let $\mathcal{C}_X$ and $\mathcal{C}_Y$ be the punctured codes of $\mathcal{C}$ by deleting the coordinates outside $X$ and $Y$, respectively.

Throughout this paper, the \textit{order} of a codeword $\vec{v}$ (or a vector) is the additive order, i.e.,
the smallest positive integer $s$ such that $s \vec{v} = \vec{0}$.
Then the subcode $\mathcal{C}_p$ is defined as
\[ \mathcal{C}_p = \left\{ \vec{c} \in \mathcal{C} : p \vec{c} = \vec{0} \right\}. \]
Let $\kappa$ be the dimension of the $p$-ary linear code $(\mathcal{C}_p)_X$.
If $\alpha = 0$, then $\kappa = 0$.
Consisting of all these parameters, we will say that $\mathcal{C}$ is of \textbf{type} $(\alpha, \beta; \gamma_0, \gamma_1, \cdots, \gamma_k; \kappa)$.

When $k = 2$, the type is always written as $(\alpha, \beta; \gamma, \delta; \kappa)$, i.e.,
the code is isomorphic to $\mathbb{Z}_p^\gamma \times \mathbb{Z}_{p^2}^\delta$.
Let $\kappa_1$ be the dimension of the linear subcode $\{(\vec{u}_1, \vec{0}) \in \mathcal{C} \}$ and $\kappa_2 = \kappa - \kappa_1$.
Similarly, it's known that the linear subcode $\{ (\vec{0}, \vec{v}_2) \in \mathcal{C}\}$ is isomorphic to $\mathbb{Z}_{p^2}^{\delta_2}$ for some $0 \leqslant \delta_2 \leqslant \delta$, then let $\delta_1 = \delta - \delta_2$.
A $\mathbb{Z}_p\mathbb{Z}_{p^2}$-additive code $\mathcal{C}$ is called \textbf{separable} if $\mathcal{C} = \mathcal{C}_X \times \mathcal{C}_Y$.
By definition, a $\mathbb{Z}_p\mathbb{Z}_{p^2}$-additive code is separable if and only if $\kappa_2 = \delta_1 = 0$; that is $\kappa = \kappa_1$ and $\delta = \delta_2$.

\begin{definition}{\cite{z2z4_additive_cyclic_codes, TY_Zhu}}\label{orign-def of cyclic}
    A subset $\mathcal{C}$ of $\mathbb{Z}_p^{\alpha} \times \mathbb{Z}_{p^k}^{\beta}$ is called a $\mathbb{Z}_p\mathbb{Z}_{p^k}$-additive cyclic code
    if it satisfies the following conditions:
    \begin{enumerate}
        \item[(1)] $\mathcal{C}$ is a $\mathbb{Z}_p\mathbb{Z}_{p^k}$-additive code;
        \item[(2)] For any codeword $\vec{u} = (\vec{u}', \vec{u}'') = (u'_0, u'_1, \cdots, u'_{\alpha - 1}, u''_0, u''_1, \cdots, u''_{\beta - 1}) \in \mathcal{C}$, its cyclic shift
        \[
            \pi(\vec{u}) = (\pi(\vec{u}'), \pi(\vec{u}'')) = (u'_{\alpha - 1}, u'_0, \cdots, u'_{\alpha - 2}, u''_{\beta - 1}, u''_0, \cdots, u''_{\beta - 2})
        \]
        is also in $\mathcal{C}$.
    \end{enumerate}
\end{definition}

\begin{remark}{\cite{z2z4_generators,ZprZps}} \label{remark: zpzp2 generator matrix 2}
    Since $\kappa = \kappa_1 + \kappa_2$ and $\delta = \delta_1 + \delta_2$,
    where $\kappa_1$ is the dimension of the subcode $\{(\vec{u}_1, \vec{0}) \in \mathcal{C}\}$
    and $\delta_2$ is the dimension of the subcode $\{ (\vec{0}, \vec{v}_2) \in \mathcal{C}\}$ whose codewords are all of order $p^2$,
    then the generator matrix can be written as:
    \begin{equation}
    \label{eq:generator matrix 2}
    \mathcal{G}_S =
        \begin{pmatrix}
            \begin{array}{cccc|ccccc}
                I_{\kappa_1} & T & T'_{p_1} & T_{p_1} & \vec{0} & \vec{0} & \vec{0} & \vec{0} & \vec{0} \\
                \vec{0} & I_{\kappa_2} & T'_{p_2} & T_{p_2} & pT_2 & pT_{\kappa_2} & \vec{0} & \vec{0} & \vec{0} \\
                \vec{0} & \vec{0} & \vec{0} & \vec{0} & pT_1 & pT'_1 & pI_{\gamma - \kappa} & \vec{0} & \vec{0} \\
                \hline
                \vec{0} & \vec{0} & S_{\delta_1} & S_p & S_{11} & S_{12} & R_1 & I_{\delta_1} & \vec{0} \\
                \vec{0} & \vec{0} & \vec{0} & \vec{0} & S_{21} & S_{22} & R_2 & R_{\delta_1} & I_{\delta_2} \\
            \end{array}
            \end{pmatrix},
    \end{equation}
    where $I_r$ is the identity matrix of size $r\times r$; $T_{p_i}, T'_{p_i}, S_{\delta_1}$ and $S_p$ are matrices over $\mathbb{Z}_p$;
    $T_1, T_2, T_{\kappa_2}, T'_1$ and $R_i$ are matrices over $\mathbb{Z}_{p^2}$ with entries in $\mathbb{Z}_p$;
    and $S_{ij}$ are matrices over $\mathbb{Z}_{p^2}$.
    Besides, $S_{\delta_1}$ and $T_{\kappa_2}$ are square matrices of full rank $\delta_1$ and $\kappa_2$, respectively.
\end{remark}
Let $\mathcal{G}$ be the generator matrix of a $\mathbb{Z}_p \mathbb{Z}_{p^2}$-additive cyclic code, then $\mathcal{G}$ and $\mathcal{G}_S$ should be permutation-equivalent,
i.e., the codes generated by $\mathcal{G}$ and $\mathcal{G}_S$ are permutation-equivalent.
Moreover, if the rows of $\mathcal{G}$ are arranged as $\mathcal{G}_S$, we call $\mathcal{G}$ is \textbf{in the form} of $\mathcal{G}_S$.
\section{Structures of \texorpdfstring{$\mathbb{Z}_p\mathbb{Z}_{p^k}$}{}-Additive Cyclic Codes} \label{sec: structures}

This section is devoted to giving the structures and the minimal spanning set of $\mathbb{Z}_p\mathbb{Z}_{p^k}$-additive cyclic codes in terms of polynomials.

\subsection{Generator Polynomials of \texorpdfstring{$\mathbb{Z}_p\mathbb{Z}_{p^k}$}{}-Additive Cyclic Codes}\label{subsec: generator polynomials of ZpZpk additive cyclic codes}

In \cite{z2z4_additive_cyclic_codes}, the authors studied the $\mathbb{Z}_2 \mathbb{Z}_4$-additive cyclic codes.
Recall that there exists a bijection between $\mathbb{Z}_p^{\alpha} \times \mathbb{Z}_{p^k}^{\beta}$ and
$R_{\alpha, \beta} = R_{\alpha} \times R_{\beta} = \mathbb{Z}_p[x]/(x^{\alpha}-1) \times \mathbb{Z}_{p^k}[x]/(x^{\beta}-1)$
given by
\[
(u'_0, \cdots, u'_{\alpha - 1}, u''_0, \cdots, u''_{\beta - 1}) \longmapsto
(u'_0 + \cdots + u'_{\alpha - 1} x^{\alpha-1}, u''_0 + \cdots + u''_{\beta - 1} x^{\beta-1}).
\]
Hence, any codeword can be regarded as a vector or as a polynomial.
This subsection is devoted to giving the structures of $\mathbb{Z}_p\mathbb{Z}_{p^k}$-additive cyclic codes in terms of polynomials.
For simplicity, $a$, $b$, $f$ always refer to the polynomials $a(x), b(x), f(x)$.
Denote the $p$-ary reduction of a polynomial $h \in \mathbb{Z}_{p^k}[x]$ by $\overline{h } \equiv h  \pmod{p}$.
Then define the following multiplication \[ h  \star (c_1, c_2) = ( \overline{h} c_1, h c_2), \]
where $(c_1, c_2) \in R_{\alpha, \beta}$ and the products of the right side are the standard polynomial products in $\mathbb{Z}_p[x]/(x^{\alpha}-1)$
and $\mathbb{Z}_{p^k}[x]/(x^{\beta}-1)$.

Let $\mathcal{C}$ be a $\mathbb{Z}_p\mathbb{Z}_{p^k}$-additive cyclic code, and we will determine the set of generators for $\mathcal{C}$
as a $\mathbb{Z}_{p^k}[x]$-submodule of $R_{\alpha, \beta} = \mathbb{Z}_p^{\alpha} \times \mathbb{Z}_{p^k}^{\beta}$.
Define the map $\Psi: \mathcal{C} \rightarrow \mathbb{Z}_{p^k}[x]/(x^{\beta}-1)$, where $\Psi((g_1, g_2)) = g_2$.
Obviously, $\Psi$ is a module homomorphism.
Then, we claim that its image is a $\mathbb{Z}_{p^k}[x]$-submodule of $\mathbb{Z}_{p^k}[x]/(x^{\beta}-1)$ and its kernel is a submodule of $\mathcal{C}$.
Moreover, the image of $\Psi$ can be regarded as an ideal in the ring $\mathbb{Z}_{p^k}[x]/(x^{\beta}-1)$.
Note that $\ker(\Psi) = \{ (a, 0) \in \mathcal{C} \mid a \in \mathbb{Z}_{p}[x]/(x^\alpha-1) \}$, then we define the set $I$ to be
\[ I = \{ a \in \mathbb{Z}_p[x]/(x^\alpha-1) \mid (a, 0) \in \ker(\Psi) \}. \]
It's clear that $I$ is an ideal and hence a cyclic code in the ring $\mathbb{Z}_p[x]/(x^\alpha-1)$.
More details about these ideals can be seen in \cite{p-adic_cyclic_codes}.

By the first isomorphism theorem and \cite[Theorem 6]{p-adic_cyclic_codes}, we have
\[ \mathcal{C}/\ker(\Psi) \cong \langle f_0, pf_1, \cdots, p^{k-1}f_{k-1} \rangle = \langle f \rangle, \]
where $f_{k-1} \mid f_{k-2} \mid \cdots \mid f_1 \mid f_0 \mid (x^\beta-1)$ and $f = f_0 + pf_1 + \cdots + p^{k-1}f_{k-1}$.

Thus, any $\mathbb{Z}_{p} \mathbb{Z}_{p^k}$-additive cyclic code can be generated as a $\mathbb{Z}_{p^k}[x]$-submodule of $R_{\alpha, \beta}$ by two elements
of the form $(a, 0)$ and $(b, f)$ for some $b\in \mathbb{Z}_p[x]/(x^\alpha-1)$,
i.e., any element in the code $\mathcal{C}$ can be described as\[ d_1 \star (a, 0) + d_2 \star (b, f), \]
where $d_1, d_2 \in \mathbb{Z}_{p^k}[x]$.
In fact, $d_1$ can be restricted in the ring $\mathbb{Z}_p[x]$.
Hence, we get \[ \mathcal{C} = \langle (a, 0), (b, f) \rangle, \] with $p$-ary polynomials $a$, $b$.

Summarize the above discussions, we have the following proposition.
\begin{proposition} \label{Zpk[x] submodule}
    The multiplication $\star$ above is well-defined and $R_{\alpha, \beta}$ is a $\mathbb{Z}_{p^k}[x]$-module with respect to this multiplication.
    Moreover, a code $\mathcal{C}$ is a $\mathbb{Z}_p\mathbb{Z}_{p^k}$-additive cyclic code if and only if $\mathcal{C}$ is a $\mathbb{Z}_{p^k}[x]$-submodule of $R_{\alpha, \beta}$.
\end{proposition}

The following theorem will give the generator polynomial of a $\mathbb{Z}_p \mathbb{Z}_{p^k}$-additive cyclic code.
\begin{theorem}
    \label{generator polynomials of cyclic}
    Let $\mathcal{C}$ be a $\mathbb{Z}_p \mathbb{Z}_{p^k}$-additive cyclic code of type $(\alpha, \beta; \gamma_0, \gamma_1, \cdots, \gamma_k; \kappa)$ with $\gcd(\beta, p) = 1$.
    Then $\mathcal{C}$ can be identified as
    \begin{enumerate}
    	\renewcommand{\labelenumi}{(\theenumi)}
    	\item $\mathcal{C} = \langle (a, 0) \rangle$, where $a \mid (x^\alpha-1) \pmod{p}$;
    	
    	\item $\mathcal{C} = \langle (b, f) \rangle$, where $f = f_0 + p f_1 + p^2 f_2 + \cdots + p^{k-1} f_{k-1}$ with $f_{k-1} \mid f_{k-2} \mid \cdots \mid f_1 \mid f_0 \mid (x^\beta-1)$;
    	
    	\item $\mathcal{C} = \langle (a, 0), (b, f) \rangle$, where $a \mid (x^\alpha-1) \pmod{p}$,
    	$f = f_0 + p f_1 + p^2 f_2 + \cdots + p^{k-1} f_{k-1}$ with $f_{k-1} \mid f_{k-2} \mid \cdots \mid f_1 \mid f_0 \mid (x^\beta-1)$ and $b$ is a polynomial over $\mathbb{Z}_p$ satisfying $\deg(b) < \deg(a)$.
    	Moreover, $a \mid h_{k-1} b$, where $f_{k-1} h_{k-1} = x^\beta - 1$.
    \end{enumerate}
\end{theorem}
\begin{proof}
	It's sufficient to show that $\deg(b) < \deg(a)$ and $a \mid h_{k-1} b$.
	\begin{itemize}
		\item Suppose $\deg(b) \geqslant \deg(a)$, then using the division algorithm (over $\mathbb{Z}_p$), we have $b = qa + r$ for some polynomials $q$ and $r$ over $\mathbb{Z}_p$, where $\deg(r) < \deg(a)$.
		Thus,
		\[ D = \langle (a, 0), (r, f) \rangle = \langle (a, 0), (b, f) \rangle = C, \]
		since $(b, f) = q \star (a, 0) + (r, f)$.
		
		\item Since $f_{k-1} \mid f_{k-2} \mid \cdots \mid f_1 \mid f_0$, then $h_{k-1} f \equiv 0 \pmod{(x^{\beta}-1)}$.
		Consider that \[ \Psi\left( \frac{x^\beta-1}{f_{k-1}} \star (b, f) \right) = \Psi\left( \left( h_{k-1} b, 0 \right) \right) = 0. \]
		Hence, \[ \left( h_{k-1} b, 0 \right)\in \ker(\Psi) \subseteq \mathcal{C} \quad \textnormal{and} \quad a \mid h_{k-1} b. \]
	\end{itemize}
	Thus, we complete the proof.
\end{proof}

\begin{example}
	Using Magma, the polynomial $x^8 - 1$ over $\mathbb{Z}_{3^2}$ can be factored as
	\[ x^8 - 1 = (x-1)(x+1)(x^2+1)(x^2 + 4x + 8)(x^2 + 5x + 8). \]
	Then, $\mathcal{C}_1 = \langle (x^2-1, 0), (x-1, (x^2+5x+8)(x^2+1)(x+1) + 3(x^2+1)(x+1)) \rangle$ can be regarded as a $\mathbb{Z}_3 \mathbb{Z}_{3^2}$-additive cyclic code.
	
	Similarly, the polynomial $x^8 - 1$ over $\mathbb{Z}_{3^3}$ can be factored as
	\[ x^8 - 1 = (x-1)(x+1)(x^2+1)(x^2 + 22x + 26)(x^2 + 5x + 26). \]
	Then, $\mathcal{C}_2 = \langle (x^2-1, 0), (x-1, (x^2+5x+26)(x^2+1)(x+1) + 3(x^2+1)(x+1) + 9(x+1)) \rangle$ can be regarded as a $\mathbb{Z}_3 \mathbb{Z}_{3^3}$-additive cyclic code.
\end{example}

\subsection{Minimal Spanning Set of \texorpdfstring{$\mathbb{Z}_p\mathbb{Z}_{p^k}$}{}-Additive Cyclic Codes}\label{subsec: minimal spanning set of ZpZpk additive cyclic codes}

In this subsection, we will give the minimal spanning set of a $\mathbb{Z}_p\mathbb{Z}_{p^k}$-additive cyclic code.

\begin{theorem}
    \label{spanning set}
    Let $\mathcal{C} = \langle (a, 0), (b, f) \rangle$ be a $\mathbb{Z}_p \mathbb{Z}_{p^k}$-additive cyclic code of type
    $(\alpha, \beta; \gamma_0, \gamma_1, \cdots, \gamma_k; \kappa)$ with $\gcd(\beta, p) = 1$,
    where $f = f_0 + pf_1  + p^2 f_2 + \cdots + p^{k-1} f_{k-1}$ with $f_{k-1} \mid f_{k-2} \mid \cdots \mid f_1 \mid f_0 \mid (x^\beta-1)$.
    Let $\deg(f_i) = t_i$, $\deg(a) = t_a$, then we define
    \[ T_1 = \bigcup_{i=0}^{\alpha-t_a-1} \left\{x^i \star (a, 0)\right\}, T_2 = \bigcup_{i=0}^{\beta-t_0-1} \left\{ x^i \star (b, f)\right\}, \]
    and \[ S_j = \bigcup_{i=0}^{t_j-t_{j+1}-1} \left\{x^i \star (\overline{h_j} b, h_j f)\right\},  \]
    where $h_j f_j \equiv x^\beta - 1 \pmod{p^k}$, $j=0, 1, \cdots, k-2$.
    Then $S = T_1 \cup T_2 \cup S_0 \cup S_1 \cup \cdots \cup S_{k-2}$ forms a minimal spanning set for $\mathcal{C}$ as a $\mathbb{Z}_{p^k}$-module.
    Moreover, we have $\gamma_0 = \alpha - t_a$, $\gamma_1 = \beta - t_0$, $\gamma_i = t_{i-2} - t_{i-1}$ for $2\leqslant i\leqslant k$.
\end{theorem}
\begin{proof}
    Let $c$ be a codeword of $\mathcal{C}$, then there exist polynomials $\lambda_1, \mu_1 \in \mathbb{Z}_{p^k}[x]$ such that
    \[ c = \overline{\lambda_1} \star \left(a, 0\right) + \mu_1 \star \left( b, f \right). \]
    \begin{enumerate}
        \item[(1)] If $\deg(\overline{\lambda_1}) \leqslant \alpha - \deg(a) - 1$, then $\overline{\lambda_1} \star (a, 0)$ can be spanned by $T_1$.
        Otherwise, by the division algorithm, we get two polynomials $\lambda_2 , \mu_2  \in \mathbb{Z}_{p^k}[x]$ such that
        \[ \overline{\lambda_1 } = \left(\frac{x^\alpha - 1}{a }\right) \overline{\lambda_2 } + \overline{\mu_2 }, \]
        where $\overline{\mu_2 } = 0$ or $\deg(\overline{\mu_2 }) \leqslant \alpha - \deg(a )-1$.
        Hence, \[ \overline{\lambda_1 } \star (a , 0) = \overline{\mu_2 } \star (a , 0). \]
        Thus, we can assume that $\overline{\lambda_1 } \star (a , 0)$ is spanned by $T_1$.

        \item[(2)] If $\deg(\mu_1) \leqslant \beta- \deg(f )-1 = \beta- t_0 - 1$, then $\mu_1  \star (b, f)$ is spanned by $T_2$.
        Otherwise, by the division algorithm, we get
        \[ \mu_1  = \left(\frac{x^\beta-1}{f_0 }\right) q_0  + r_0  = h_0 q_0  + r_0 , \]
        where $r_0  = 0$ or $\deg(r_0 ) \leqslant \beta - t_0 - 1$, which implies that
        \[ \mu_1  \star (b , f ) = q_0  \star (\overline{h_0} b , h_0 f ) + r_0  \star (b , f ). \]
        Obviously, $r_0  \star (b , f )$ is spanned by $T_2$.
        Thus, we only need to consider $q_0  \star (\overline{h_0} b , h_0 f )$.

        \item[(3)] The degree of $q_0 $ is important. If $\deg(q_0 ) \leqslant t_0-t_1-1$, then we are done.
        Otherwise, by division algorithm again, we have \[ q_0  = \left(\frac{x^\beta-1}{h_0 f_1 }\right) q_1  + r_1  = \frac{h_1}{h_0} q_1 + r_1, \]
        for some $q_1 , r_1  \in \mathbb{Z}_{p^k}[x]$ with $r_1  = 0$ or $\deg(r_1) \leqslant t_0 - t_1 -1$.
        Thus, \[ q_0 \star (\overline{h_0} b, h_0 f) = q_1 \star (\overline{h_1} b, h_1 f) + r_1 \star (\overline{h_0} b, h_0 f), \]
        which implies that $r_1 \star (\overline{h_0} b, h_0 f)$ is spanned by $S_0$.

        \item[(4)] As for $q_1$, we have similar discussion.
        For general, if $\deg(q_j ) \leqslant t_j-t_{j+1}-1$ for some $0\leqslant j\leqslant k-3$, then we are done.
        Otherwise, by division algorithm again, we have
        \[ q_j  = \left(\frac{x^\beta-1}{h_j f_{j+1} }\right) q_{j+1}  + r_{j+1}  = \frac{h_{j+1}}{h_j} q_{j+1} + r_{j+1}, \]
        for some $q_{j+1} , r_{j+1}  \in \mathbb{Z}_{p^k}[x]$ with $r_{j+1}  = 0$ or $\deg(r_{j+1}) \leqslant t_j - t_{j+1} -1$.
        Thus, \[ q_j \star (\overline{h_j} b, h_j f) = q_{j+1} \star (\overline{h_{j+1}} b, h_{j+1} f) + r_{j+1} \star (\overline{h_j} b, h_j f), \]
        which implies that $r_{j+1} \star (\overline{h_j} b, h_j f)$ is spanned by $S_j$.
        Note that $h_{k-2}  f  = p^{k-1} h_{k-2} f_{k-1} $ but $h_{k-1} f = h_{k-1} f_{k-1} = x^\beta-1 = 0$.
        Therefore, if $j = k-2$, then
        \[ q_{k-2} \star (\overline{h_{k-2}}b, h_{k-2}f) = q_{k-1} \star (\overline{h_{k-1}}b, 0) + r_{k-1} \star (\overline{h_{k-2}} b, h_{k-2} f). \]
        Since $a  \mid \frac{x^\beta-1}{f_{k-1} } b  = h_{k-1}  b $, then $q_{k-2} \star (\overline{h_{k-2}}b, h_{k-2}f)$ is spanned by $T_1$ and $S_{k-2}$.
    \end{enumerate}
    Therefore, $S = T_1 \cup T_2 \cup S_0 \cup S_1 \cup \cdots \cup S_{k-2}$ is a spanning set.
    Moreover, no element in $S$ is a linear combination of the other elements, so $S$ is minimal. On the other hand, the order of the codewords spanned by $S_j$ is just $p^{k-1-j}$.
    Then $\gamma_0 = \alpha - t_a$, $\gamma_1 = \beta - t_0$, and $\gamma_i = t_{i-2} - t_{i-1}$ for $2\leqslant i\leqslant k$.
\end{proof}

\begin{remark}
    If $k=2$, we get $\mathcal{C} = \langle (a, 0), (b, fh + pf) \rangle$ with $fhg = x^\beta - 1$, where $a,b\in \mathbb{Z}_p[x]$ and $f,g,h\in \mathbb{Z}_{p^2}[x]$.
    Moreover, if $p=2$, then Theorem \ref{generator polynomials of cyclic} and Theorem \ref{spanning set} are actually \cite[Theorem 11]{z2z4_additive_cyclic_codes}
    and \cite[Theorem 13]{z2z4_additive_cyclic_codes}, respectively.
\end{remark}

Straightforward from the proof \cite[Theorem 5 \& Proposition 6]{z2z4_generators} and Theorem \ref{spanning set}, it's easy to get the following result.
\begin{proposition}
	\label{thm: type and generator polynomials}
	Let $\mathcal{C} = \langle (a, 0), (b, fh + pf) \rangle$ be a $\mathbb{Z}_p \mathbb{Z}_{p^2}$-additive cyclic code of type
	$(\alpha, \beta; \gamma, \delta = \delta_1 + \delta_2; \kappa = \kappa_1 + \kappa_2)$ with $fhg = x^\beta - 1$ and $\gcd(\beta, p) = 1$.
	Then
	\[ \begin{aligned}
		& \gamma = \alpha - \deg(a) + \deg(h), \quad \delta = \deg(g), \quad \kappa = \alpha - \deg(\gcd(a, b\overline{g})), \\
		& \kappa_1 = \alpha - \deg(a), \quad \kappa_2 = \deg(a) - \deg(\gcd(a, b\overline{g})),
	\end{aligned} \]
	where $\overline{g} \equiv g \pmod{p}$.
\end{proposition}

\begin{example}
	Let $\mathcal{C} = \langle (x^2-1, 0), (x-1, (x^2+5x+8)(x^2+1)(x+1) + 3(x^2+1)(x+1)) \rangle$ be a $\mathbb{Z}_3 \mathbb{Z}_9$-additive cyclic code with $\alpha = 4$ and $\beta = 8$.
	Here, $f = (x^2+1)(x+1)$, $h = x^2+5x+8$ and $g = (x-1)(x^2+4x+8)$, then its type is $(4, 8; 4, 3; 3)$.
	Moreover, by \Cref{spanning set}, the minimal spanning set of $\mathcal{C}$ is the set of polynomials corresponding to the rows of the following matrix:
	\[
	\left(
	\begin{array}{cccc|cccccccc}
		2 & 0 & 1 & 0 & 0 & 0 & 0 & 0 & 0 & 0 & 0 & 0 \\
		0 & 2 & 0 & 1 & 0 & 0 & 0 & 0 & 0 & 0 & 0 & 0 \\
		\hline
		2 & 1 & 0 & 0 & 2 & 7 & 8 & 8 & 6 & 1 & 0 & 0 \\
		0 & 2 & 1 & 0 & 0 & 2 & 7 & 8 & 8 & 6 & 1 & 0 \\
		0 & 0 & 2 & 1 & 0 & 0 & 2 & 7 & 8 & 8 & 6 & 1 \\
		\hline
		0 & 0 & 1 & 2 & 3 & 6 & 6 & 0 & 6 & 3 & 3 & 0 \\
		2 & 0 & 0 & 1 & 0 & 3 & 6 & 6 & 0 & 6 & 3 & 3 \\
	\end{array}
	\right).
	\]
	For instance, $(2100 \ 27886100)$ means $(x-1, x^5 + 6x^4 + 8x^3 + 8x^2 + 7x + 2)$.
	Moreover, this matrix can be regarded as a generator matrix of $\mathcal{C}$.
\end{example}

For a linear code over $\mathbb{Z}_p$, it's call \textit{optimal} if its parameter $[n, k, d]$ is the same as the best linear code (BKLC) in the database \cite{codetable}.
And, the code with parameter $[n, k, d]$ is called \textit{almost optimal} if the $[n, k, d+1]$ code exists and is optimal.
Many optimal and almost optimal codes are obtained by the Gray images, and the parameters are listed in \Cref{table: parameters of p-ary images,table: parameters of p-ary images of ZpZp2}.
Furthermore, the Gray images of many cyclic codes over $\mathbb{Z}_{p^2}$ obtained from $x^n-$ are nonlinear.
	
In \cite{codetable}, when fixed the length $n$, the dimension $k$, the best known linear codes have indirect, multi-step constructions in many cases.
Therefore, it's desirable to obtain them by another method, for example, in the form of the Gray images of cyclic codes over $\mathbb{Z}_{p^2}$ and $\mathbb{Z}_p \mathbb{Z}_{p^2}$.

However, notice that this paper is devoted to giving the structures for the cyclic codes whose Gray images are linear over $\mathbb{Z}_p$, rather than searching for new BKLCs.

\begin{table}[!t]
	\renewcommand{\arraystretch}{1.3}
	\caption{Gray images of cyclic codes over $\mathbb{Z}_p \mathbb{Z}_{p^2}$}
	\label{table: parameters of p-ary images of ZpZp2}
	\centering
	\begin{threeparttable}
		\begin{tabular}{c|c|c|c}
			\hline\hline
			$p$ & $(\alpha, \beta)$ & \textnormal{Generators} & \textnormal{Gary Image} \\
			\hline\hline
			$p=3$ &
			$(2, 5)$ & \tabincell{c}{$f = x^4 + x^3 + x^2 + x + 1$, $h=1$, \\
				$g=x-1$, $a = 0$, $b = x + 1$} & $[17, 2, 12]^{\dagger}$  \\
			\hline
			$p=3$ &
			$(8, 4)$ & \tabincell{c}{$f = x+1$, $h=x^2+1$, $g=x-1$, \\
				$a = 0$, $b = x^7 + 2x^6 + x^5 + x$} & $[20, 4, 12]^{\dagger}$ \\
			\hline
			$p=3$ &
			$(5, 8)$ & \tabincell{c}{$f = x^5 + 5x^4 + 4x^3 + 4x^2 + 3x+ 8$, $h=x^2 + 5x + 8$, \\ $g=x-1$, $a = 0$, $b = x^3 + x^2 + 2x + 2$} & $[29, 4, 18]^{\dagger}$ \\
			\hline
			$p=3$ &
			$(6, 8)$ & \tabincell{c}{$f = x^5 + 5x^4 + 4x^3 + 4x^2 + 3x+ 8$, $h=x^2 + 5x + 8$, \\ $g=x-1$, $a = 0$, $b = x^3 + 2x^2 + x + 1$} & $[30, 4, 19]^{\dagger}$ \\
			\hline\hline
			$p=5$ &
			$(3, 4)$ & \tabincell{c}{$f =x^2 + 19x + 18$, $h=x + 7$, $g=x-1$, \\
				$a = 0$, $b = 2x + 1$} & $[23, 3, 18]^{\dagger}$  \\
			\hline
			$p=5$ &
			$(6, 4)$ & \tabincell{c}{$f =x^2 + 19x + 18$, $h=x + 7$, $g=x-1$, \\
				$a = 0$, $b = x^3 + 2x^2 + 3x + 4$} & $[26, 3, 20]^{\dagger}$  \\
			\hline
			$p=5$ &
			$(5, 8)$ & \tabincell{c}{$f =x^6 + 8x^5 + 7x^4 + x^2 + 8x + 7$, $h=x + 7$, \\
				$g=x-1$, $a = 0$, $b = x^3 + 3x + 4$} & $[45, 3, 35]^{\dagger}$  \\
			\hline
			$p=5$ &
			$(6, 8)$ & \tabincell{c}{$f =x^6 + 8x^5 + 7x^4 + x^2 + 8x + 7$, $h=x + 7$, \\
				$g=x-1$, $a = 0$, $b = x^4 + 2x^3 + x^2 + 1$} & $[46, 3, 36]^{\dagger}$  \\
			\hline\hline
			$p=7$ &
			$(0, 6)$ & \tabincell{c}{$f=x^4 + 32x^3 + 13x^2 + 12x + 30$, \\
				$h = x+18$, $g = x-1$, 	$a=b=0$} & $[42, 3, 35]^{\dagger}$ \\
			\hline
			$p=7$ &
			$(4, 6)$ & \tabincell{c}{$f = x^4 + 32x^3 + 13x^2 + 12x + 30$, $h=x+18$, \\
				$g=x-1$, $a = 0$, $b = x^2 + 2x + 1$} & $[46, 3, 39]^{\dagger}$ \\
			\hline\hline
		\end{tabular}
		\begin{tablenotes}
			\item[$\dagger$] means the code is optimal.
			\item[$\ddagger$] means the code is almost optimal.
		\end{tablenotes}
	\end{threeparttable}
\end{table}

\begin{table}[!t]
	\renewcommand{\arraystretch}{1.3}
	\caption{Gray images of cyclic codes over $\mathbb{Z}_{p^2}$}
	\label{table: parameters of p-ary images}
	\centering
	\begin{threeparttable}
		\begin{tabular}{c|c|c|c|c|c}
			\hline\hline
			$p$ & $n$ & $f$ & $h$ & $g$ & \textnormal{Gray Image}  \\
			\hline\hline
			$p=3$ & $4$ & $(x+1)$ & $(x^2+1)$ & $(x-1)$ & $[12, 4, 6]^{\dagger}$  \\
			\hline
			$p=3$ & $5$ & $(x^4+x^3+x^2+x+1)$ & $1$ & $(x-1)$ & $[15, 2, 10]^{\ddagger}$  \\
			\hline
			$p=3$ & $7$ & $(x^6+x^5+x^4+x^3+x^2+x+1)$ & $1$ & $(x-1)$ & $[21, 2, 14]^{\ddagger}$  \\
			\hline
			$p=3$ & $8$ & $(x+1)(x^2+1)(x^2+4x+8)$ & $(x^2+5x+8)$ & $(x-1)$ & $[24, 4, 15]^{\dagger}$  \\
			\hline
			$p=3$ & $8$ & $(x+1)(x^2+1)(x^2+5x+8)$ & $(x^2+4x+8)$ & $(x-1)$ & $[24, 4, 15]^{\dagger}$  \\
			\hline\hline
			$p=5$ & $4$ & $(x+1)(x+18)$ & $(x+7)$ & $(x-1)$ & $[20, 3, 15]^{\dagger}$  \\
			\hline
			$p=5$ & $4$ & $(x+1)(x+7)$ & $(x+18)$ & $(x-1)$ & $[20, 3, 15]^{\dagger}$ \\
			\hline
			$p=5$ & $4$ & $(x+1)(x+7)(x+18)$ & $1$ & $(x-1)$ & $[20, 2, 16]^{\dagger}$ \\
			\hline
			$p=5$ & $7$ & $(x^6 + x^5 + x^4 + x^3 + x^2 + x + 1)$ & $1$ & $(x-1)$ & $[35, 2, 28]^{\ddagger}$ \\
			\hline
			$p=5$ & $8$ & $(x+1)(x+7)(x^2+7)(x^2-7)$ & $(x-7)$ & $(x-1)$ & $[40, 3, 30]^{\ddagger}$ \\
			\hline
			$p=5$ & $8$ & $(x+1)(x-7)(x^2+7)(x^2-7)$ & $(x+7)$ & $(x-1)$ & $[40, 3, 30]^{\ddagger}$ \\
			\hline
			$p=5$ & $8$ & $(x^7 + x^6 + x^5 + x^4 + x^3 + x^2 + x + 1)$ & $1$ & $(x-1)$ & $[40, 2, 32]^{\ddagger}$  \\
			\hline\hline
		\end{tabular}
		\begin{tablenotes}
			\item[$\dagger$] means the code is optimal.
			\item[$\ddagger$] means the code is almost optimal.
		\end{tablenotes}
	\end{threeparttable}
\end{table}

\section{\texorpdfstring{$p$}{}-ary Images of \texorpdfstring{$\mathbb{Z}_p\mathbb{Z}_{p^2}$}{}-Additive Cyclic Codes}
\label{sec:Gray image of cyclic codes}

In \cite{binary_images_of_z2z4_cyclic} and \cite{binary_images_of_z4_cyclic}, the authors studied the Gray images of the cyclic codes over $\mathbb{Z}_4$ and $\mathbb{Z}_2 \mathbb{Z}_4$, respectively.
This section is devoted to generalizing the results about Gray images to $\mathbb{Z}_{p^2}$ and $\mathbb{Z}_{p} \mathbb{Z}_{p^2}$, respectively.
In \cite{wan1_Z4-code}, the Gray map $\Phi$ defined for quaternary codes satisfies $\Phi(\vec{u} + \vec{v}) = \Phi(\vec{u}) + \Phi(\vec{v}) + \Phi(2 \vec{u} \ast \vec{v})$,
where $\ast$ denotes the \textbf{componentwise product} of two vectors.
As for the codes over $\mathbb{Z}_{p^2}$ and $\mathbb{Z}_p \mathbb{Z}_{p^2}$, we have a similar result.
Now, we reprove \cite[Lemma 2]{LXX} by \cite[Lemma 2]{wsk_z3z9_linear_rank_kernel} and a slightly different proof.

\begin{lemma}\label{phi(u+v) = phi(u)+phi(v)+...}
    Let $\vec{u} = (u_1, \cdots, u_{\alpha + \beta}), \vec{v} = (v_1, \cdots, v_{\alpha + \beta}) \in \mathbb{Z}_p^{\alpha} \times \mathbb{Z}_{p^2}^{\beta}$, then
    \begin{equation} \label{eq: Phi(u+v)=Phi(u)+Phi(v)+Phi(pP(u,v))}
        \Phi(\vec{u} + \vec{v}) = \Phi(\vec{u}) + \Phi(\vec{v}) + \Phi(pP(\vec{u}, \vec{v})),
    \end{equation}
    where $P(\vec{u}, \vec{v}) = \left( P(u_1, v_1), \cdots, P(u_{\alpha+\beta}, v_{\alpha+\beta}) \right)$ and
    \begin{equation}\label{eq: P(u,v)}
    	P(u_i, v_i) \equiv \sum_{a=1}^{p-1} \sum_{b=p-a}^{p-1} \frac{\prod_{m=0}^{p-1} (u_i-m)(v_i-m)}{(u_i-a)(v_i-b)} \pmod{p} =
    	\begin{cases}
    		1, & \overline{u}_i + \overline{v}_i \geqslant p, \\
    		0, & \overline{u}_i + \overline{v}_i < p, \\
    	\end{cases}
    \end{equation}
	with $\overline{u}_i \equiv u_i \pmod{p}$, $\overline{v}_i \equiv v_i \pmod{p}$.
\end{lemma}
\begin{proof}
    It's sufficient to show the values of $P(u_i, v_i)$ in \eqref{eq: P(u,v)}.
    Note that \[ \sum_{a=1}^{p-1} \sum_{b=p-a}^{p-1} \frac{\prod_{m=0}^{p-1} (x-m)(y-m)}{(x-a)(y-b)} \equiv \sum_{a+b \geqslant p} \frac{(x^p-x)(y^p-y)}{(x-a)(y-b)} \pmod{p}, \]
    since $x^p - x \equiv \prod_{m=0}^{p-1} (x-m) \pmod{p}$ when $p$ is prime.
    For $0 \leqslant x, y \leqslant p-1$, if $x \neq a$ or $y \neq b$, then \[ \frac{(x^p-x)(y^p-y)}{(x-a)(y-b)} \equiv 0 \pmod{p}. \]
    Then, according to the famous Wilson's Theorem: $(p-1)! \equiv -1 \pmod{p}$, we have
    \[ \sum_{a+b \geqslant p} \frac{(x^p-x)(y^p-y)}{(x-a)(y-b)} \equiv \sum_{\substack{a+b \geqslant p,\\x=a,y=b}} \frac{(x^p-x)(y^p-y)}{(x-a)(y-b)} \pmod{p} = \begin{cases}
    	1, & x + y \geqslant p, \\
    	0, & x + y < p. \\
    \end{cases} \]
	The equality \eqref{eq: P(u,v)} holds since $P(x, y) = P(\bar{x}, \bar{y})$ when $x, y \in \mathbb{Z}_{p^2}$.
\end{proof}

\begin{lemma} \label{lemma: Phi(C) is linear iff pP(u,v) in C}
    Let $\mathcal{C}$ be a $\mathbb{Z}_p\mathbb{Z}_{p^2}$-additive code.
    The $\mathbb{Z}_p\mathbb{Z}_{p^2}$-linear code $C = \Phi(\mathcal{C})$ is linear if and only if $pP(\vec{u},\vec{v}) \in \mathcal{C}$ for all $\vec{u},\vec{v} \in \mathcal{C}$.
\end{lemma}
\begin{proof}
    Assume that $C$ is linear, then for any $\vec{u}, \vec{v} \in \mathcal{C}$, we have  $\Phi(\vec{u}), \Phi(\vec{v}), \Phi(\vec{u} + \vec{v}) \in C$.
    Thus, by Lemma \ref{phi(u+v) = phi(u)+phi(v)+...}, $\Phi(pP(\vec{u}, \vec{v})) \in C$.
    Considering that $\phi(px) = \phi(py)$ if and only if $px = py$, then $pP(\vec{u}, \vec{v}) \in \mathcal{C}$, which is important and occurs many times in this paper.

    Conversely, for every $\vec{x}, \vec{y}\in C$, there exist $\vec{u}', \vec{v}' \in \mathcal{C}$ such that $\vec{x} = \Phi(\vec{u}')$ and $\vec{y} = \Phi(\vec{v}')$.
    Thus, $pP(\vec{u}', \vec{v}') \in \mathcal{C}$, which implies that $\vec{u}' + \vec{v}' + (p-1)pP(\vec{u}', \vec{v}') \in \mathcal{C}$.
    Hence, by Lemma \ref{phi(u+v) = phi(u)+phi(v)+...},
    \[ \Phi(\vec{u}' + \vec{v}' + (p-1)pP(\vec{u}', \vec{v}')) = \Phi(\vec{u}' + \vec{v}') + (p-1)\Phi(pP(\vec{u}', \vec{v}')) \in C. \]
    Thus,
    \[
        \begin{aligned}
            \vec{x} + \vec{y} & = \Phi(\vec{u}') + \Phi(\vec{v}') = \Phi(\vec{u}') + \Phi(\vec{v}') + p \Phi(pP(\vec{u}', \vec{v}')) \\
            & = \Phi(\vec{u}' + \vec{v}') + (p-1)\Phi(pP(\vec{u}', \vec{v}')) \in C,
        \end{aligned}
    \]
    and we complete the proof.
\end{proof}

\begin{remark}
	It's important that $\Phi(pP(\vec{u}, \vec{v})) \in \Phi(\mathcal{C})$ if and only if $pP(\vec{u}, \vec{v}) \in \mathcal{C}$.
	Recall the characteristic function of $S$ defined in \cite[Section IV]{LS_Gray_map} as:
	\[  f(x,y) = \left( (x^p - x)(y^p - y) \sum_{i=1}^{p-1} \frac{1}{x-i} \prod_{(i,j)\in S} \frac{1}{y-j} \right) ^{p-1}.  \]
	where $S$ is the set of the ordered pairs $(i,j) \in \mathbb{Z}_p^2$ with $i + j \geqslant p$.
	Then $f(x, y) \equiv \left( P(x, y) \right)^{p-1} \pmod{p}$.
	Furthermore, let $L = \{1+2i | 1 \leqslant i \leqslant (p-1)/2 \} \cup \{ p-1 \}$ by $L$, which gives the degrees of terms of $P(x,y)$ by \cite[Lemma 4]{LXX}.
\end{remark}

\subsection{\texorpdfstring{$p$}{}-ary Images of Cyclic Codes over \texorpdfstring{$\mathbb{Z}_{p^2}$}{}}\label{subsec:Gray image of Zp2 cyclic codes}

Note that in \cite{LS_Gray_map}, linear cyclic codes over $\mathbb{Z}_{p^2}$ whose Gray images are linear cyclic codes over $\mathbb{Z}_p$ are characterized.
In this subsection, we will give the necessary and sufficient condition of the Gray image to be linear (not necessarily cyclic) of some cyclic code over $\mathbb{Z}_{p^2}$.

%

In this section, the elements of the cyclic code $\mathcal{C} = \langle f h  + pf  \rangle$ over $\mathbb{Z}_{p^2}$ are always regarded as polynomials,
where $fhg = x^n - 1$ and $\gcd(n,p)=1$.
Moreover, it's easy to check that $\mathcal{C} = \langle fh, pfg \rangle = \langle fh, pf \rangle$ (see \Cref{3 generators of zpzp2 cyclic codes} or \cite[Proposition 5]{binary_images_of_z4_cyclic} for binary case), since $h$ and $g$ are coprime.
Then the subcode $\mathcal{C}_p$ of $\mathcal{C}$ whose codewords are of order $p$ (and the zero polynomial is in $\mathcal{C}_p$) is also a cyclic code,
and $\mathcal{C}_p = \langle pf \rangle$ (see \Cref{generators of C_p} or \cite[Proposition 6]{binary_images_of_z4_cyclic} for binary case).

Recall some results in \cite{book_of_MacWilliams_in_1977}.
Let $\ast$ denote the \textbf{componentwise product} of two vectors.
For convenience, the componentwise product of $i$ vectors $\vec{u}$ is written as $\vec{u}^i$.
Moreover, the componentwise product $a \ast b$ of two polynomials $a$ and $b$ is defined as: \[ a \ast b = \left(\sum_{i=0}^m a_ix^i\right) \ast \left(\sum_{i=0}^m b_ix^i\right) = \sum_{i=0}^m a_i b_i x^i. \]
And the componentwise product of two sets $A, B$ is defined as: $A \ast B = \{ a\ast b | a\in A, b\in B \}$.

From now on, denote the principal ideal generated by $f \in \mathbb{Z}_p[x]/(x^n-1)$ with $\gcd(n,p)=1$ by $\langle f \rangle_p^n$, then we have

\begin{proposition}
	Let $\mathcal{C} = \langle fh + pf \rangle$ be a cyclic code over $\mathbb{Z}_{p^2}$ of length $n$ with $fhg  = x^n-1$ and $\gcd(n, p) = 1$,
	then for every positive integer $i$, we have
	\[ pI^i = p \underbrace{\langle fh \rangle \ast \dots \ast \langle fh \rangle}_{i} = \{ p \underbrace{v_1 \ast \dots \ast v_i}_i | v_1, \dots, v_i \in \mathcal{C} \}, \]
	and
	\[ I_i = \underbrace{\langle \overline{fh} \rangle_p^n \ast \dots \ast \langle \overline{fh} \rangle_p^n}_{i} = \{ \underbrace{\overline{v_1} \ast \dots \ast \overline{v_i}}_i | v_1, \dots, v_i \in \mathcal{C} \}. \]
\end{proposition}
\begin{proof}
	The equalities hold since for every polynomial (codeword) $v \in \mathcal{C}$, $pfh \mid pv$ and $\overline{fh} \mid \overline{v}$.
\end{proof}

\begin{proposition} \label{prop: I_2+I_m in <f>}
    Let $\mathcal{C} = \langle fh + pf \rangle$ be a cyclic code over $\mathbb{Z}_{p^2}$ of length $n$ with $fhg  = x^n-1$ and $\gcd(n, p) = 1$,
    then the following statements are equivalent:
    \begin{enumerate}
    	\item[(1)] $\Phi(\mathcal{C})$ is a $p$-ary linear code,
    	\item[(2)] $pP(\overline{r_1 fh}, \overline{r_2 fh}) \in \langle pf \rangle$ for any $r_1, r_2 \in \mathbb{Z}_{p^2}[x]$,
    	\item[(3)] $pI^{\ell} \subseteq \langle pf \rangle$ (or equivalently, $I_{\ell} \subseteq \langle \overline{f} \rangle_p^n$) for every $\ell \in L = \{1+2i | 1 \leqslant i \leqslant (p-1)/2 \} \cup \{ p-1 \}$.
    \end{enumerate}
\end{proposition}
\begin{proof}
	First, we prove $(1)$ and $(2)$ are equivalent.
    It's easy to check that $pm_1 m_2 = p\overline{m}_1\overline{m}_2$ if $m_1, m_2\in \mathbb{Z}_{p^2}$, where $\overline{m}_1 \equiv m_1 \pmod{p}$ and $\overline{m}_2 \equiv m_2 \pmod{p}$.
    It's known that  $\Phi(\mathcal{C})$ is linear if and only if for any $\vec{u}, \vec{v}\in \mathcal{C}$, $pP(\overline{\vec{u}}, \overline{\vec{v}}) \in \mathcal{C}$,
    where $\overline{\vec{u}} \equiv \vec{u} \pmod{p}$ and $\overline{\vec{v}} \equiv \vec{v} \pmod{p}$.
    Since $\mathcal{C}$ is a cyclic code, then $\vec{u} = r_1 G$, $\vec{v} = r_2 G$ for some polynomials $r_1$, $r_2 \in \mathbb{Z}_{p^2}[x]$, where $G = fh + pf$.
    Hence, $pP(\overline{r_1 G}, \overline{r_2 G}) \in \mathcal{C}_p$, i.e., $pP(\overline{r_1 fh}, \overline{r_2 fh}) \in \langle pf \rangle$.

    As for $(1)$ and $(3)$, note that $I^{\ell}$ is the set of all the componentwise products of $\ell$ codewords.
    If $(3)$ holds, then $(2)$ holds, hence $\Phi(\mathcal{C})$ is linear.
    On the other hand, if $(1)$ holds, then by \cite[Theorem 2]{LXX}, $\Phi(pI^{\ell})$, $\ell \in L$ are generators of $\Phi(\mathcal{C})$.
    Equivalently, $pI^{\ell} \subseteq \langle pf \rangle$.
%
\end{proof}

Let $c$ be a divisor of $x^n-1 \in \mathbb{F}_p[x]$, then $\left( c\otimes c \otimes \cdots \otimes c \right)_j$ is defined to be the polynomial whose roots are the products $\xi^{i_1} \xi^{i_2} \cdots \xi^{i_j}$, where $j$ is a positive integer and $\xi^{i_1}, \xi^{i_2}, \cdots, \xi^{i_j}$ are roots of $c$.
For convenience, denote $\left( c\otimes c \otimes \cdots \otimes c \right)_j$ by $c_j$, which is called the $j_{th}$ \textbf{circle product} of $c$.
For example, $c_2 = (c \otimes c)_2$ and $c_3 = (c \otimes c \otimes c)_3$.
In general, $c_j \mid x^n - 1$ since $c \mid x^n-1$.
Equivalently, the roots of $c_j$ form a subset (not necessarily a subgroup) of the group of the $n$-th roots of unity.
In addition, $0 \leqslant \deg(c_i) \leqslant \deg(c_{i+1}) \leqslant n$, $i \geqslant 1$.

Denote the set of roots of $x^n-1$ by $U_n$, then the set of roots of $c_j$ is a subset of $U_n$.

\begin{example}
	It's known that the roots of $c = x^2 + 1 = (x-\xi)(x-\xi^{-1}) \in \mathbb{Z}_3[x]$ are in the extension field $\mathbb{F}_9$.
	Let $S = \{1, -1\}$, then
	\[ S+S = \{-2, 0, 2\}, \quad S+S+S = \{-3, -1, 1, 3 \}. \]
	That's to say, the roots of $c_2$ and $c_3$ are
	\[ S_1 = \{ \xi^{-2}, \xi^0, \xi^2 \}=\{1,-1\} \subseteq U_4, \quad \textnormal{and} \quad S_2 = \{\xi^{-3}, \xi^{-1}, \xi, \xi^{3}\}=\{\xi, \xi^{-1}\} \subseteq U_4, \]
	respectively.
	Hence, $c_2 = x^2-1$ and $c_3 = x^2+1$.
\end{example}

\begin{lemma} \label{lemma: J_2+J_m in <e_m>}
    Let $c$ and $d$ be two polynomials in $\mathbb{Z}_p[x]$ such that $x^n - 1 = c d$, where $\gcd(n, p) = 1$.
    Denote the $j_{th}$ circle product of $c$ by $c_j$ and let $e_j$ be the polynomial such that $c_j e_j = x^n-1$ for $j=2, \cdots, p$.
    Then $J_i \subseteq \langle e_i \rangle_p^n$,
    where $J = J_1 = \langle d \rangle_p^n$ and $J_{i+1} = J_i \ast J$ for $i=1,2,\cdots,p-1$.
\end{lemma}
\begin{proof}
    Let $s_i(x) = a_1^i(x)d(x) \ast a_2^i(x)d(x) \ast \cdots \ast a_i^i(x)d(x) \in J_i$
    for $i=1,2,\cdots,p$ and $S_i(x)$ be the Mattson-Solomon polynomials of $s_i$.
    Then, by \cite[Ch.8.\S 6. Theorem 22]{book_of_MacWilliams_in_1977}, we have
    \[ S_i(x) = \sum_{k=0}^{n-1} s_i(z^k) x^k = \frac{1}{n^{i-1}} A_1^i(x) A_2^i(x) \cdots A_i^i(x), \]
    where $z$ is a primitive $n$-th root of unity in some extension field of $\mathbb{Z}_p$,   $A_j^i(x)$ is the Mattson-Solomon polynomials of $a_j^i(x)d(x)$ with $1\leqslant j\leqslant i$ and
    \[ A_j^i(x) = \sum_{k=0}^{n-1} a_j^i(z^k)d(z^k) x^k. \]
    Since $x^n-1 = d(x)c(x)$, then $d(z^k) \neq 0$ if and only if $c(z^k) = 0$.
    Thus,
    \[
        \begin{aligned}
            S_i(x) & = \frac{1}{n^{i-1}} \left( \sum_{k_1=0}^{n-1} a_1^i(z^{k_1})d(z^{k_1}) x^{k_1} \right) \cdots
                        \left( \sum_{k_i=0}^{n-1} a_i^i(z^{k_i})d(z^{k_i}) x^{k_i} \right) \\
            & = \frac{1}{n^{i-1}} \left( \sum_{c(z^{k_1})=0} a_1^i(z^{k_1})d(z^{k_1}) x^{k_1} \right)
            \cdots \left( \sum_{c(z^{k_i})=0} a_i^i(z^{k_i})d(z^{k_i}) x^{k_i} \right) \\
            & = \frac{1}{n^{i-1}} \left( \sum_{c_i(z^k)=0} \sum_{k_1+ \cdots + k_i = k} a_1^i(z^{k_1})d(z^{k_1}) \cdots a_i^i(z^{k_i})d(z^{k_i}) x^k \right).
        \end{aligned}
    \]
    It follows that if $s_i(z^k) \neq 0$, then $c_i(z^k) = 0$ for some $0\leqslant k\leqslant n-1$.
    Equivalently, if $e_i(z^k) = 0$, then $s_i(z^k) = 0$.
    This means $e_i(x)$ divides $s_i(x)$ and $s_i(x)\in \langle e_i(x) \rangle_p^n$.
    Thus, $J_i \subseteq \langle e_i \rangle_p^n$.
\end{proof}

It's clear that two polynomials over a field are coprime (over some extension of $\mathbb{Z}_p$) if and only if their roots are totally distinct.
Using this simple property, we give the following lemma.

\begin{lemma} \label{lemma: lemma of Zp2 linear}
	Keep the notations in \Cref{prop: I_2+I_m in <f>} and \Cref{lemma: J_2+J_m in <e_m>}.
	If $\overline{f}$ is coprime with $\overline{g_i}$, then $I_i \subseteq \langle \overline{f} \rangle_p^n$.
	Conversely, if there exists a polynomial $s_i \in I_i \subseteq \langle \overline{f} \rangle_p^n$ such that $s_i(z^j) \neq 0$ for all
	$j \in K_i = \{ k | \overline{g_i}(z^k) = 0, 0\leqslant k\leqslant n-1 \}$, then $\overline{f}$ is coprime with $\overline{g_i}$,
	where $z$ is a primitive $n$-th root of unity in some extension field of $\mathbb{Z}_p$, and $\overline{g_i}$ is the $i_{th}$ circle product of $\overline{g}$, $i \geqslant 1$.
\end{lemma}
\begin{proof}
	By \Cref{prop: I_2+I_m in <f>} and \cite[Theorem 4.13]{LS_Gray_map}, the case $\overline{f} = 1$ or $\overline{g} = 1$ is trivial.
	From now on, assume that $\overline{f} \neq 1$ and $\overline{g_i} \neq 1$.
	Let $e_i$ be the polynomial such that $g_i e_i = \overline{g_i} \overline{e_i} = x^n-1$.
	
	The first statement follows from \Cref{lemma: J_2+J_m in <e_m>}.
	So it's sufficient to consider the last statement.
	
	Obviously, $K_i \neq \emptyset$ since $1 \neq \overline{g} \mid x^n-1$.
	Let $s_i(x) = \overline{\lambda_i}(x) \overline{f}(x) \in I_i \subseteq \langle \overline{f} \rangle_p^n$ be the polynomial that satisfies the condition.
	Using the Mattson-Solomon transform, we obtain
	\[ S_i(x) = \sum_{j=0}^{n-1} s_i(z^j) x^j = \sum_{j=0}^{n-1} \overline{\lambda_i}(z^j) \overline{f}(z^j) x^j. \]
	Then for every $k \in K_i$, $\overline{\lambda_i}(z^k) \overline{f}(z^k) \neq 0$.
	It follows that $\overline{f}(z^k) \neq 0$.
	We deduce that $\overline{g_i}(z^k) = 0$ implies that $\overline{f}(z^k) \neq 0$.
	Equivalently, if $\overline{f}(z^k) = 0$, then $\overline{g_i}(z^k)\neq 0$. Namely, $\overline{e_i}(z^k) = 0$.
	Therefore, $\overline{f} $ divides $\overline{e_i} $ and $\overline{f}$ is coprime with $\overline{g}_i$.
\end{proof}

By the above lemmas, we give the necessary and sufficient condition on the linearity of the Gray image of some cyclic code over $\mathbb{Z}_{p^2}$.

\begin{theorem} \label{thm: linearity of images of Zp2 cyclic codes}
	Let $\mathcal{\mathcal{C}} = \langle fh  + pf \rangle$ be a linear cyclic code over $\mathbb{Z}_{p^2}$ of length $n$ with $fhg = x^n-1$ and $\gcd(n, p) = 1$.
	Then $\Phi(\mathcal{\mathcal{C}})$ is linear if and only if $\overline{f}$ is coprime with $\overline{g_\ell}$ for all $\ell \in L$,
	where $\overline{g_\ell}$ is the $\ell_{th}$ circle product of $\overline{g}$.
\end{theorem}
\begin{proof}
	The sufficiency follows from \Cref{prop: I_2+I_m in <f>} and \Cref{lemma: lemma of Zp2 linear}.
	
	Assume that $\Phi(\mathcal{C})$ is linear, then $I_{\ell} \subseteq \langle \overline{f} \rangle_p^n$ for every $\ell \in L$.
	Let $s_i(x) = s_{i-1}(x) \ast x^{m_i}d(x) \in I_i$ and $S_i(x)$ be the Mattson-Solomon polynomials of $s_i$,
	where $s_{i-1}(x) \in I_{i-1}$ and $s_1(x) = d(x)$, $d = \overline{fh}$, and $0 \leqslant m_i \leqslant n-1$.
	Then for $i \geqslant 2$, by \cite[Ch.8.\S 6. Theorem 22]{book_of_MacWilliams_in_1977}, we have
	\[ S_i(x) = \sum_{k \in K_i} s_i(z^k) x^k = \frac{1}{n} \left( \sum_{t \in K_{i-1}} s_{i-1}(z^t) x^t \right)\left( \sum_{j\in K_1} z^{j m_i}d(z^j) x^j \right),   \]
	where $z$ is a primitive $n$-th root of unity in some extension field of $\mathbb{Z}_p$.
	
	We claim that there exists $0 \leqslant m_i \leqslant n-1$ such that $s_i(z^k) \neq 0$ for every $k \in K_i$.

	Use induction on $i$.
	The case $i = 1$ holds trivially, since $s_1(z^k) = z^0 d(z^k) \neq 0$ for every $k \in K_1$.
	Then assume that $s_{i-1}(z^t) \neq 0$ for every $t \in K_{i-1}$.
	For every $k \in K_i$, define
	\[ \Gamma_{k}^{i}(x) = \frac{1}{n} \sum_{j\in J_i(k)} s_{i-1}(z^{k-j}) d(z^{j}) x^{j}, \]
	where \[ J_i(k) = \{ j | d(z^j) \neq 0, s_{i-1}(z^{k-j}) \neq 0 \}. \]
	If $k \in K_i$ and $j \in K_1$, then $k-j \in K_{i-1}$, which implies that $s_{i-1}(z^{k-j}) \neq 0$ and $J_i(k) \neq \emptyset$.
	Since $s_{i-1}(z^{k-j}) d(z^{j}) \neq 0$ for every $j \in K_1$, then $\Gamma_k^{i}$ is a nonzero polynomial whose degree is at most $n-1$.
	Then there exists an integer $0 \leqslant m_{i} \leqslant n-1$ such that $\Gamma_{k}^{i}(z^{m_{i}}) = s_{i}(z^k) \neq 0$.
	Thus, the claim holds.
	
	From the claim, $s_{\ell} \in I_{\ell}$ is the polynomial such that for every $j \in K_{\ell}$, $s_{\ell}(z^j) \neq 0$.
	By \Cref{lemma: lemma of Zp2 linear}, $\overline{f}$ is coprime with $\overline{g_\ell}$,
	where $\ell \in L$.
\end{proof}

\begin{remark}\label{remark: linearity Zp2}
	The Gray images of the codes are called \textit{trivial} in which $f=1$, or $g=1$, or $g=x^m-1$ for some $m \mid n$,
	since the Gray images are always linear.
	Assume that $f \neq 1$ and let $\Gamma$ be the set of roots of $\overline{g}$.
	It's known that $\overline{g_i}$ is the polynomial whose roots are just in the set $\Gamma^i$ for some integer $i$, where
	\[ \Gamma^i = \Gamma \cdots \Gamma = \{ \xi_1 \dots \xi_i | \xi_j \in \Gamma, 1 \leqslant j \leqslant i \} \subseteq U_n. \]
	It follows that $ab \in \Gamma^{i+j}$ if $a \in \Gamma^i$ and $b \in \Gamma_j$.
	Therefore, the nontrivial linear Gray image exists if $\overline{f} \mid (1+x^{e}+x^{2e} + \dots + x^{n-e})$, where $e$ is the order of $\overline{g}$ and $e < n$.
	Furthermore, $\overline{f} \mid (1+x^{e}+x^{2e} + \dots + x^{n-e})$ if and only if $\overline{f}$ is coprime with $\overline{g_i}$ for every positive integer $i$.
\end{remark}

In \Cref{sec: examples}, for some special $n$, the linear Gray images of the cyclic codes obtained from $x^n-1$ are all trivial.
Thus, it's a difficult but interesting topic to find the nontrivial linear Gray image as optimal codes.

\subsection{\texorpdfstring{$p$}{}-ary Images of \texorpdfstring{$\mathbb{Z}_p\mathbb{Z}_{p^2}$}{}-Additive Cyclic Codes}\label{subsec:Gray image of ZpZp2 cyclic codes}
In this subsection, we will give a necessary and sufficient condition of the $p$-ary image of a $\mathbb{Z}_p\mathbb{Z}_{p^2}$-additive cyclic code to be linear.

\begin{lemma} \label{lemma: linearity of Gray images with separable C_p}
    About $\Phi(\mathcal{C}_Y)$, we have the following two assertions:
    \begin{enumerate}
        \item[(1)] Let $\mathcal{C}$ be a $\mathbb{Z}_p\mathbb{Z}_{p^2}$-additive code such that $\Phi(\mathcal{C})$ is linear, then $\Phi(\mathcal{C}_Y)$ is also linear.
        \item[(2)] Let $\mathcal{C}$ be a $\mathbb{Z}_p\mathbb{Z}_{p^2}$-additive code of length $n$ with $\gcd(n, p) = 1$.
        If $\mathcal{C}_p$ is separable, then $\Phi(\mathcal{C})$ is linear if and only if $\Phi(\mathcal{C}_Y)$ is linear.
    \end{enumerate}
\end{lemma}
\begin{proof}
    As for $(1)$, for any $\vec{u} = (u, u'), \vec{v} = (v, v') \in \mathcal{C}$, $pP(\vec{u}, \vec{v}) = (\vec{0}, pP(u', v')) \in \mathcal{C}_p\subseteq \mathcal{C}$,
    which means $pP(u', v')\in \mathcal{C}_Y$.
    Thus, $\Phi(\mathcal{C}_Y)$ is linear. As for $(2)$, we only need to consider the sufficiency.
    If $\Phi(\mathcal{C}_Y)$ is linear and $u', v'\in \mathcal{C}_Y$, then $pP(u', v')\in \mathcal{C}_Y$, which means there exists $\vec{w} = (w, pP(u', v')) \in \mathcal{C}$.
    Thus, $pP(u', v')\in (\mathcal{C}_p)_Y$.
    Hence, $(\vec{0}, pP(u', v'))\in \mathcal{C}_p \subseteq \mathcal{C}$ since $\mathcal{C}_p$ is separable and $\vec{0} \in (\mathcal{C}_p)_X$.
\end{proof}

\begin{lemma} \label{3 generators of zpzp2 cyclic codes}
    Let $\mathcal{C} = \langle (a, 0), (b, fh + pf)\rangle$ be a $\mathbb{Z}_p\mathbb{Z}_{p^2}$-additive cyclic code with $fhg = x^\beta-1$.
    Then $\mathcal{C}$ can be also generated by $(a, 0)$, $(b\overline{g}, pfg)$ and $(b', fh)$, where $b' = b - \overline{\mu}b\overline{g}$ and $\lambda h + \mu g = 1$.
\end{lemma}
\begin{proof}
    Similar to the proof of \cite[Lemma 10]{binary_images_of_z2z4_cyclic}.
\end{proof}

\begin{lemma}
    \label{generators of C_p}
    Let $\mathcal{C} = \langle (a, 0), (b, fh + pf)\rangle$ be a $\mathbb{Z}_p\mathbb{Z}_{p^2}$-additive cyclic code with $fhg = x^\beta-1$, then the subcode
    $\mathcal{C}_p$ is generated by $(a , 0), (b  \overline{g }, pf  g )$ and $(0, pf  h )$.
\end{lemma}
\begin{proof}
    According to the definition, $\mathcal{C}_p \subseteq \mathcal{C}$ is the subcode which contains all codewords of order $p$.
    Obviously, the order of any codeword $\vec{u} \in \langle (a , 0), (b  \overline{g }, pf  g ), (0, pf  h ) \rangle$ is $p$.

    Note that $\mathcal{C} = \mathcal{C}_p$ if and only if $\deg(g) = 0$, i.e., $g = 1$ by \Cref{thm: type and generator polynomials}.
    In this case. $fh = x^\beta - 1 = 0 = pfh$ when regarded as vectors.
    Let $\lambda = 0$ and $\mu = 1$, then $b' = b-b = 0$ and the result follows from \Cref{3 generators of zpzp2 cyclic codes}.

    From now on, assume that $\mathcal{C} \neq \mathcal{C}_p$.
    Let $\vec{0} \neq \vec{v}\in \mathcal{C}_p$, then there exist polynomials $d_1, d_2\in \mathbb{Z}_{p^2}[x]$ such that
    \[ \vec{v} = \overline{d_1} \star (a, 0) + d_2 \star (b, f h + pf), \quad p\vec{v} = (0, pd_2 fh) = 0. \]
    Using the division algorithm, $d_2 = qg + d'_2$ for some $q$ and $d'_2$, where $\deg(d'_2) < \deg(g)$.
    Then $pd_2fh \equiv pd'_2fh \pmod{(x^\beta-1)}$ and $pd_2fh \equiv pd'_2fh \pmod{p^2}$.
    Since $pd'_2 fh$ is a polynomial with degree smaller than $\beta$, then $pd'_2 fh \equiv 0 \pmod{p^2}$.
    If $fh \equiv 0 \pmod{p}$, then $p(b, fh + pf) = 0$, which implies that $\mathcal{C} = \mathcal{C}_p$, contradiction.
    Thus, we have $d'_2 \equiv 0 \pmod{p}$ and $d'_2 = pk $ for some nonzero polynomial $k$.
    Hence, $d_2 = qg + pk$ and
    \begin{align*}
    	\vec{v} & = \overline{d_1} \star (a, 0) + (qg + pk) \star (b, f h + pf) \\
    	& = \overline{d_1} \star (a, 0) + q \star (b\overline{g}, pfg) + k \star (0, fh).
    \end{align*}
    Therefore, $\vec{v} \in \langle (a, 0), (b \overline{g}, p f g), (0, pf h) \rangle$ and we conclude the proof.
\end{proof}

From \cite{ZprZps} and the matrix \eqref{eq:generator matrix 2} in \Cref{remark: zpzp2 generator matrix 2}, the generator matrix $\mathcal{G}_p$ of $\mathcal{C}_p$ has the form:
\begin{equation}
    \label{eq: generator matrix of C_p}
    \mathcal{G}_p =
    \begin{pmatrix}
        \begin{array}{ccc|ccc}
            I_{\kappa_1} & A & B & \vec{0} & \vec{0} & \vec{0} \\
            \vec{0} & I_{\kappa_2} & C & pT'_2 & \vec{0} & \vec{0} \\
            \vec{0} & \vec{0} & \vec{0} & pT_1 & pI_{\gamma - \kappa} & \vec{0} \\
            \vec{0} & \vec{0} & \vec{0} & pS & pR & pI_\delta \\
        \end{array}
    \end{pmatrix}.
\end{equation}
Comparing Lemma \ref{3 generators of zpzp2 cyclic codes} and Lemma \ref{generators of C_p},
the subcodes generated by $(b', fh)$ and the last $\delta$ rows of \eqref{eq:generator matrix 2} in \Cref{remark: zpzp2 generator matrix 2} should be permutation-equivalent.
Then any codeword generated by the last $\delta_2$ rows should belong to $\langle (b', fh) \rangle$, up to equivalence.

\begin{lemma} \label{lemma: claim of the ZpZp2 linearity thm}
    Let $\mathcal{C} = \langle (a, 0), (b, fh + pf)\rangle$ be a $\mathbb{Z}_p\mathbb{Z}_{p^2}$-additive cyclic code with $fhg = x^\beta-1$.
    Let $\mathcal{C}'$ be the subcode generated by $(b', fh)$ and $S = \{ \vec{u} = (\vec{0}, u') \in \mathcal{C} \}$,
    where $b' = b - \overline{\mu}b\overline{g}$ and $\lambda h + \mu g = 1$.
    Then $\mathcal{C}'$ is cyclic and \[ \mathcal{C}' = \left\langle (b', fh), \left(0, \frac{pfga}{\gcd(a, b\overline{g})} \right) \right\rangle. \]
\end{lemma}
\begin{proof}
	Denote by $\mathcal{I}$ the ideal generated by the two generators.
	Note that
	\begin{equation}
		\label{eq: expression of generators in C'}
		\left(0, \frac{pfga}{\gcd(a, b\overline{g})}\right) =
		- \frac{b \overline{g}}{\gcd(a, b\overline{g})} \star (a,0) + \frac{a}{\gcd(a,b\overline{g})}\star (b\overline{g},pfg) \in\mathcal{C}',
	\end{equation}
	since $a,b\in \mathbb{Z}_p[x]/(x^{\alpha}-1)$.
	Then $\mathcal{I} \subseteq \mathcal{C}'$ and it's sufficient to show $S \subseteq \mathcal{I}$.
	
    By Lemma \ref{3 generators of zpzp2 cyclic codes}, every codeword $(0, c)$ can be written as:
    \begin{equation}
        \label{eq: expression of codewords in C'}
        (0, c) = a_1 \star (a, 0) + a_2 \star (b\overline{g}, pfg) + a_3 \star (b', fh),
    \end{equation}
    for some $a_1, a_2, a_3\in \mathbb{Z}_{p^2}[x]/(x^\beta-1)$, then we have the following two cases.

    The order of $(0, c)$ is $p$ if and only if $\overline{a_1}a + \overline{a_2}b\overline{g} = 0$,
    and $\overline{a_1} \frac{a}{\gcd(a, b\overline{g})} = -\overline{a_2} \frac{b\overline{g}}{\gcd(a, b\overline{g})}$.
    Hence, $\frac{a}{\gcd(a, b\overline{g})} \mid \overline{a_2}$ since $\frac{a}{\gcd(a, b\overline{g})}$ and $\frac{b\overline{g}}{\gcd(a, b\overline{g})}$ are coprime.
    Thus, $\overline{a_2} = \overline{d}\frac{a}{\gcd(a, b\overline{g})}$ for some $\overline{d}$.
    Therefore, \[ a_1 \star (a, 0) + a_2 \star (b\overline{g}, pfg) = (0, a_2 pfg) = d \star \left(0, \frac{pfga}{\gcd(a, b\overline{g})} \right) \in \mathcal{I}, \]
    where $d$ is the Hensel lift of $\overline{d}$.

    If the order of $(0,c)$ is $p^2$, by \eqref{eq:generator matrix 2} in Remark \ref{remark: zpzp2 generator matrix 2}, we can get $(0, c) = (0, c_1) + (0, c_2)$,
    where the order of $(0, c_1)$ is $p$ and $(0, c_2) \in \langle (b', fh) \rangle$.
    By the above proof, $(0, c_1) \in \mathcal{I}$.
    Therefore, $S \subseteq \mathcal{I}$ and $\mathcal{C}' = \mathcal{I}$ is cyclic.
\end{proof}

Next we will give a necessary and sufficient condition for the Gray image to be linear.

\begin{theorem} \label{thm: linearity of images of ZpZp2 cyclic codes}
    Let $\mathcal{C} = \langle (a, 0), (b, fh + pf)\rangle$ be a $\mathbb{Z}_p\mathbb{Z}_{p^2}$-additive cyclic code with $fhg = x^\beta-1$.
    Then $\Phi(\mathcal{C})$ is linear if and only if $\overline{f'} = \overline{f}a / \gcd(a, \overline{b}g)$ is coprime with $\overline{g_\ell}$ for every $\ell \in L$.
\end{theorem}
\begin{proof}
    Let $\vec{u} = (u, u')$, $\vec{w} = (w, w')\in \mathcal{C}$.
    Obviously, if $\vec{u} \in \mathcal{C}_p$, then $pP(\vec{u}, \vec{w}) = \vec{0} \in \mathcal{C}$.
    Note that by Lemma \ref{3 generators of zpzp2 cyclic codes} and Lemma \ref{generators of C_p},
    $\mathcal{C} = \langle (a, 0), (b\overline{g}, pfg) (b', fh) \rangle$ and $\mathcal{C}_p = \langle (a, 0), (b\overline{g}, pfg) (0, pfh)\rangle$,
    where $\lambda h + \mu g = 1$.
    Then, we only need to consider the codewords that are not in $\mathcal{C}_p$ and $(b', fh)$ will always be one generator of such codewords.
    Since $pP(\vec{u}, \vec{w}) = (\vec{0}, pP(u', w'))$, we define the set $S = \{ \vec{u} = (\vec{0}, u') \in \mathcal{C} \}$
    and the subcode $\mathcal{C}' = \langle (b', fh), S\rangle \subseteq \mathcal{C}$.
    Comparing with \eqref{eq:generator matrix 2} in \Cref{remark: zpzp2 generator matrix 2}, the generator matrix of $\mathcal{C}'$ is (permutation-equivalent) in the following form:
    \begin{equation} \label{eq:generator matrix of C'}
        \left(  \begin{array}{ccc|ccc}
            \vec{0} & \vec{0} & \vec{0} & p T_1 & p I_{\gamma - \kappa} & \vec{0} \\
            \vec{0} & \vec{0} & S' & S'' & R & I_{\delta} \\
        \end{array} \right).
    \end{equation}

    For any codeword $\vec{u} \notin \mathcal{C}'$, by \eqref{eq:generator matrix 2} in \Cref{remark: zpzp2 generator matrix 2}, its order is $p$, then $pP(\vec{u}, \vec{w}) = (\vec{0}, \vec{0}) \in \mathcal{C}'$.
    Thus, we only need to consider the codewords $\vec{u}, \vec{w}\in \mathcal{C}'$.
    Obviously, $pP(\vec{u}, \vec{w}) = (\vec{0}, pP(u', w')) \in S \subseteq \mathcal{C}' \subseteq \mathcal{C}$.
    Therefore, $\Phi(\mathcal{C})$ is linear if and only if $pP(\vec{u}, \vec{w}) \in \mathcal{C}'$ for any $\vec{u}, \vec{w}\in \mathcal{C}'$,
    which means, for all $\vec{u}, \vec{w}\in \mathcal{C}'$, $pP(\vec{u}, \vec{w}) \in \mathcal{C}'$, i.e., $\Phi(\mathcal{C}')$ is linear.

    By Lemma \ref{lemma: claim of the ZpZp2 linearity thm}, $\mathcal{C}'$ is cyclic and
    \[ \mathcal{C}' = \left\langle (b', fh), \left(0, \frac{pfga}{\gcd(a, b\overline{g})} \right) \right\rangle. \]

    Let $f' = \frac{fa}{\gcd(a, b\overline{g})}$, $h' = \frac{h\gcd(a, b\overline{g})}{a}$ and $g' = g$.
    It's clear that $f'g'h' = x^\beta-1$ with $f'$, $g'$ and $h'$ coprime factors.
    Thus, we have $\mathcal{C}' = \langle (b', f'h'), (0, pf')\rangle$.
    Obviously, the codewords of order $p$ in $\mathcal{C}'$ belong to $S$, i.e., the left part is all zero vector.
    Then, $\mathcal{C}'_p = \langle (0, pf'h'), (0, pf')\rangle = \langle (0, pf')\rangle$ is separable since $\langle (0, pf')\rangle = \{ 0\} \times \langle pf'\rangle$.
    By Lemma \ref{lemma: linearity of Gray images with separable C_p}, $\Phi(\mathcal{C}')$ is linear if and only if $\Phi(\mathcal{C}'_Y)$ is linear.
    Moreover, $\mathcal{C}'_Y = \langle f'h', pf' \rangle = \langle f'h' + pf' \rangle$ since $f'h'g' = x^\beta - 1$.
    According to Theorem \ref{thm: linearity of images of Zp2 cyclic codes}, $\Phi(\mathcal{C}'_Y) = \Phi(\langle f'h' + pf' \rangle)$ is linear if and only if
    $\gcd(\overline{f'}, \overline{g'_\ell}) = 1$ for every $\ell \in L$.
\end{proof}

\begin{remark}
    Let $\mathcal{C} = \langle (a, 0), (b, fh + pf)\rangle$ be a $\mathbb{Z}_p\mathbb{Z}_{p^2}$-additive cyclic code with $fhg = x^\beta-1$.
    Similar to the proof of \cite[Theorem 12]{binary_images_of_z2z4_cyclic}, it's easy to obtain that $\Phi(\mathcal{C})$ is linear if $g = 1$ or $g = x^s-1$ with $s\mid \beta$.
   And these linear Gray images are called \textit{trivial}.
   Similar to \Cref{remark: linearity Zp2}, the nontrivial linear Gray image exists if
   \[ \gcd\left( \frac{a}{\gcd(a, b\overline{g})}, x^e-1 \right) = 1, \quad \textnormal{and} \quad \overline{f} \mid (1+x^e+x^{2e}+\dots+x^{n-e}), \]
   where $e$ is the order of $\overline{g}$ and $e < n$.
\end{remark}

\begin{example}
	Let $\mathcal{C} = \langle (x-1,0), (1,fh + 3f) \rangle$ be a $\mathbb{Z}_3 \mathbb{Z}_9$-additive cyclic code of type $(\alpha \geqslant 2, 11; \gamma, \delta; \kappa)$.
	Since $x^{11}-1 = (x-1) (x^5+7x^4-x^3+x^2+6x-1) (x^5+3x^4-x^3+x^2+2x-1)$, then let $p_5 = x^5+7x^4-x^3+x^2+6x-1$ and $q_5 = x^5+3x^4-x^3+x^2+2x-1$.
	Both of $p_5$ and $q_5$ are basic irreducible polynomials over $\mathbb{Z}_9$.
	Then the linearity of $\Phi(\mathcal{C})$ is listed in Table \ref{table: Gray image of Z3Z9 cyclic codes-3}.
\end{example}

\begin{table}[h]
	\caption{Linearity of $\Phi(\mathcal{C})$ with $\mathcal{C} = \langle (x-1,0), (1,fh + 3f) \rangle$:}
	\label{table: Gray image of Z3Z9 cyclic codes-3}
	\centering
	\begin{tabular}{c|c|c|c}
		\hline\hline
		$f$ & $g$ & $h$ & Linearity \\
		\hline\hline
		$-$ & $1$ & $-$ & Linear \\
		\hline
		$-$ & $x-1$ & $-$ & Linear \\
		\hline
		$p_5$ & $q_5$ & $x-1$ & Nonlinear \\
		\hline
		$1$ & $q_5$ & $(x-1)p_5$ & Linear \\
		\hline
		$x-1$ & $q_5$ & $p_5$ & Nonlinear \\
		\hline
		$q_5$ & $p_5$ & $x-1$ & Nonlinear \\
		\hline
		$x-1$ & $p_5$ & $q_5$ & Nonlinear \\
		\hline
		$1$ & $p_5$ & $(x-1)q_5$ & Linear\\
		\hline
		$1$ & $(x-1)p_5$ & $q_5$ & Linear \\
		\hline
		$q_5$ & $(x-1)p_5$ & $1$ & Nonlinear \\
		\hline
		$1$ & $(x-1)q_5$ & $p_5$ & Linear \\
		\hline
		$p_5$ & $(x-1)q_5$ & $1$ & Nonlinear \\
		\hline\hline
	\end{tabular}
\end{table}

\section{Linearity of Gray Images of Some Special Cyclic Codes over \texorpdfstring{$\mathbb{Z}_{p^2}$}{}}\label{sec: examples}

Let $\langle fh + pf \rangle$ be a cyclic code over $\mathbb{Z}_{p^2}$ and $fhg = x^n - 1$.
By Theorem \ref{thm: linearity of images of Zp2 cyclic codes}, the linearity of the Gray image are determined completely.
In this section, assume that $q$ is prime and $p \neq q$.

\begin{example} \label{thm: examples}
	Let $n = q^2$ with the prime $q$ and $x^n - 1 = (x-1)ab$,
	where $\overline{a} = 1 + x + \cdots + x^{q-1}$ and $\overline{b} = 1 + x^q + x^{2q} + \cdots + x^{(q-1)q}$ are both irreducible over $\mathbb{F}_p$.
	
	From now on, assume that $f \neq 1$, $g \neq 1$ and $g \neq x-1$.
	It's easy to check that $\overline{a} \otimes \overline{a} = x^q-1$ and $\overline{b} \otimes \overline{b} = x^{q^2} - 1$, which implies $\Phi(\mathcal{C})$ is linear only if $g \mid (x-1)a$.
	
	Note that $\overline{g_2} = \dots = \overline{g_p} = (x-1)\overline{a}$ whenever $g = a$ or $g = (x-1)a$.
	Thus, $\Phi(\mathcal{C})$ is linear if and only if $\overline{f}$ is coprime with $(x-1)\overline{a}$, i.e., $f = b$,
	since the roots of $\overline{b}$ are not $q_{th}$ roots of unity.
	
	Then the linearity of the Gray image of the code $\mathcal{C} = \langle fh, pf \rangle$ is
	listed in \Cref{table: examples 1}.
\end{example}

\begin{table}[h]
	\caption{Linearity of the codes in \Cref{thm: examples}}
	\label{table: examples 1}
	\centering
	\begin{tabular}{c|c|c|c|c}
		\hline\hline
		Length & $f$ & $g$ & $h$ & Linearity  \\
		\hline\hline
		Any & $1$ & $\ast$ & $\ast$ & Linear    \\
		\hline
		Any & $\ast$ & $1$ & $\ast$ & Linear  \\
		\hline
		Any & $\ast$ & $x-1$ & $\ast$ & Linear   \\
		\hline
		$q^2$ & $a$ & $(x-1)b$ & $1$ & Nonlinear    \\
		\hline
		$q^2$ & $a$ & $b$ & $x-1$ & Nonlinear     \\
		\hline
		$q^2$ & $b$ & $(x-1)a$ & $1$ & Linear     \\
		\hline
		$q^2$ & $b$ & $a$ & $x-1$ & Linear     \\
		\hline
		$q^2$ & $x-1$ & $a$ & $b$ & Nonlinear   \\
		\hline
		$q^2$ & $x-1$ & $b$ & $a$ & Nonlinear     \\
		\hline
		$q^2$ & $x-1$ & $ab$ & $1$ & Nonlinear    \\
		\hline
		$q^2$ & $(x-1)a$ & $b$ & $1$ & Nonlinear    \\
		\hline
		$q^2$ & $(x-1)b$ & $a$ & $1$ & Nonlinear     \\
		\hline\hline
	\end{tabular}
\end{table}

\begin{table}[!t]
    \renewcommand{\arraystretch}{1.3}
    \caption{Linearity of $\Phi(\mathcal{C})$ with $\mathcal{C} = \langle (a,0), (b,fh + pf) \rangle$}
    \label{table: Gray image of Z3Z9 cyclic codes-1}
    \centering
    \begin{tabular}{c|c|c|c|c|c}
    \hline\hline
    $\alpha$ & $\beta$ & $f$ & $g$ & $h$ & Linearity \\
    \hline\hline
    - & - & - & $1$ & - & Linear \\
    \hline
    - & - & $1$ & $ - $ & - & Linear \\
    \hline
    - & - & - & $x-1$ & - & Linear \\
    \hline
    - & $ts$ & - & $x^s-1$ & - & Linear \\
    \hline\hline
    \end{tabular}
\end{table}

\begin{example}
    Let $\mathcal{C} = \langle (a,0), (b,fh + pf) \rangle$ be a $\mathbb{Z}_3 \mathbb{Z}_9$-additive cyclic code of type $(\alpha, \beta; \gamma, \delta; \kappa)$.
    Then the linearity of $\Phi(\mathcal{C})$ of some special $\mathcal{C}$ are listed in Table \ref{table: Gray image of Z3Z9 cyclic codes-1}.
\end{example}

Let $n$ and $m$ be two positive integers such that $\gcd(n,m) = 1$,
then the least positive integer $k$ for which $m^{k} \equiv 1 \pmod{n}$ is called the \textit{multiplicative order} of $m$ modulo $n$.

Assume that $n$ and $p$ are coprime, then the \textit{cyclotomic coset} of $p$ modulo $n$ containing $i$ is defined by
\[ C_i = \{ (i\cdot p^j \pmod{n}) \in \mathbb{Z}_n : j=0,1,\cdots \}. \]

\begin{theorem} \label{thm: factors of x^n-1}
	Let $n$ be odd prime and $\ell = (n-1)/2$ be the multiplicative order of $p$ modulo $n$, then the following statements hold:
	
	$(1)$ over $\mathbb{Z}_{p^2}$, we have $x^n - 1 = (x-1) f_1 f_2$,
	where $\overline{f_1}$ and $\overline{f_2}$ are two irreducible polynomials of degree $\ell$ over $\mathbb{F}_p$;
	
	$(2)$ $\overline{f}_i \otimes \overline{f}_i = \frac{x^n -1}{x-1}$, $\overline{f}_i \otimes \overline{f}_i \otimes \overline{f}_i = x^n -1$ and $\gcd(\overline{f_i}, \overline{f_j} \otimes \overline{f_j}) = \overline{f_i}$, where $i=1,2$;
	
	$(3)$ the Gray image of the cyclic code $\langle fh + pf \rangle$ over $\mathbb{Z}_{p^2}$ is linear if and only if $f=1$ or $g\in \{ 1, x-1 \}$, where $fgh = (x-1) f_1 f_2 = x^n - 1$.
\end{theorem}
\begin{proof}
	By \cite[Theorem 2.45 \& Theorem 2.47]{Finite_Fields}, it's known that $x^n - 1 = (x-1) \overline{f_1} \overline{f_2}$,
	where $\overline{f_1}$ and $\overline{f_2}$ are two irreducible polynomials of degree $\ell$ over $\mathbb{F}_p$.
	Then, using Hensel's lemma, we can get $x^n - 1 = (x-1) f_1 f_2$ and the first statement $(1)$ holds.
	
	As for $(2)$, $\mathbb{Z}_n$ (actually, $\mathbb{Z}_n$ is a field) can be partitioned into three parts, $C_0$, $C_1$ and $C_i$ for some $i$
	and the sizes of $C_1$ and $C_i$ are both $\ell$.
	Obviously, $C_1 = \{ 1, p, p^2, \cdots, p^{\ell-1} \}$, since $\ell$ is the multiplicative order of $p$ modulo $n$.
	Note that $p$ is the square (\textit{quadratic residue modulo $n$}) of some element in $\mathbb{Z}_n^{\ast} = \mathbb{Z}_n \setminus \{ 0 \}$.
	Thus, all the elements in $C_1$ are all the squares of $\mathbb{Z}_n^{\ast}$.
	Equivalently, $C_i$ is the set of all non-squares (\textit{quadratic non-residue modulo $n$}) in $\mathbb{Z}_n^{\ast}$.
	
	Recall the \textit{quadratic character} $\eta$ (see \cite[Example 5.10]{Finite_Fields}) on $\mathbb{Z}_n^{\ast}$ and the \textit{integer-valued function} $\upsilon$
	(see \cite[Definition 6.22]{Finite_Fields}) of $\mathbb{Z}_n$, where
	\[
	\eta(x) =
	\begin{cases}
		1, & x \ \textnormal{is a square}, \\
		-1, & x \ \textnormal{is a non-square}, \\
	\end{cases}
	\quad
	\textnormal{and}
	\quad
	\upsilon(y) =
	\begin{cases}
		-1, & y \in \mathbb{Z}_n^{\ast}, \\
		n-1, & y = 0.  \\
	\end{cases}
	\]
	Let $N(a_1 x_1^2 + a_2 x_2^2 = x_0)$ be the number of solutions of the equation $a_1 x_1^2 + a_2 x_2^2 = x_0$ in $\mathbb{Z}_n^2$, where $a_1, a_2 \in \mathbb{Z}_n^{\ast}$.
	Thus, for every $x_0 \in \mathbb{Z}_n$, by \cite[Lemma 6.24]{Finite_Fields}, we have
	\[ N(a_1 x_1^2 + a_2 x_2^2 = x_0) = n + \upsilon(x_0) \eta(-a_1 a_2) \geqslant n-(n-1) = 1 > 0. \]
	Hence, every element in $\mathbb{Z}_n^{\ast}$ is the sum of two squares, and also the sum of two non-squares.
	Note that $N(a_1 x_1^2 + a_2 x_2^2 = 0) = $
	Therefore, \[ C_1 + C_1 = C_i + C_i = \mathbb{Z}_n^{\ast}, \]
	where the sum $A+B$ of two sets $A$ and $B$ is the set $\{ a+b | a\in A, b\in B \}$.
	Therefore, the roots of $\overline{f}_1$ are also the roots of $\overline{f}_2 \otimes \overline{f}_2$, whose roots are $\mathbb{Z}_n^{\ast}$.
	
	Moreover, it's easy to check that
	\[ C_1 + C_1 + C_1 = C_i + C_i + C_i = \mathbb{Z}_n, \]
	which implies that $\overline{f}_i \otimes \overline{f}_i \otimes \overline{f}_i = x^n - 1$, $i=1, 2$.
	
	As for $(3)$, it's easy to check that $\overline{g_3} = x^n - 1$ when $g \neq 1$ and $g \neq x-1$.
	Hence, the Gray image is linear if and only if $f=1$.
\end{proof}

\begin{remark}
	If $p=2$, then the item $(3)$ of \Cref{thm: factors of x^n-1} will be a little different.
	It should be modified as \textit{the Gray image is linear if and only if $f=1$ or $g\in \{ 1, x-1 \}$, or $f=x-1$, $g=f_i$ for some $i$}.
	For example, over $\mathbb{Z}_4$, we have \[ x^7 - 1 =  (x + 3) (x^3 + 2x^2 + x + 3) (x^3 + 3x^2 + 2x + 3). \]
	Let $f = x-1$, $g = x^3 + 2x^2 + x + 3$, $h = x^3 + 3x^2 + 2x + 3$, then the Gray image of the cyclic code $\langle fh + 2f \rangle$ is linear.
\end{remark}

\section{Conclusions}\label{sec:conclusions}
In this paper, we studied the $p$-ary images (Gray images) of cyclic codes over $\mathbb{Z}_{p^2}$ and $\mathbb{Z}_p \mathbb{Z}_{p^2}$.
The main contributions of this paper are the following:
\begin{enumerate}
	\renewcommand{\labelenumi}{(\theenumi)}
	\item We defined the polynomial $P(x_1, x_2)$, a characteristic function with lower degree than $f(x,y)$ defined in \cite[Secction IV]{LS_Gray_map}.
	Note that the degrees of terms of $P(x_1,x_2)$ plays an important role in this paper.
	
	\item We gave the generator polynomials of a $\mathbb{Z}_p \mathbb{Z}_{p^k}$-additive cyclic code (See \Cref{generator polynomials of cyclic}), which is a generalization of \cite[Theorem 11]{z2z4_additive_cyclic_codes}.
	
    \item We obtained the minimal spanning set of a $\mathbb{Z}_p \mathbb{Z}_{p^k}$-additive cyclic code (See Theorem \ref{spanning set}) and when $k = 2$, we gave the relation between its type and generator polynomials (See \Cref{thm: type and generator polynomials}).
    And some optimal codes (See \Cref{table: parameters of p-ary images,table: parameters of p-ary images of ZpZp2}) are found by different ways from the methods mentioned in \cite{codetable}.

    \item We presented a necessary and sufficient condition of the Gray image of $\mathbb{Z}_p \mathbb{Z}_{p^2}$-additive cyclic code to be linear
            (See Theorem \ref{thm: linearity of images of Zp2 cyclic codes} and Theorem \ref{thm: linearity of images of ZpZp2 cyclic codes}).
    We determined some special $\mathbb{Z}_p \mathbb{Z}_{p^2}$-additive cyclic codes with $\gcd(\beta, p) = 1$ whose Gray images are linear $p$-ary codes.
    Specially, if $\alpha = 0$, they are the linear cyclic codes over $\mathbb{Z}_{p^2}$ whose Gray images are linear.
            As for the linear cyclic Gray image, see \cite[Theorem 4.13]{LS_Gray_map}.

    \item We determined the linearity of the Gray images for two classes of cyclic codes (See \Cref{thm: examples} and \Cref{thm: factors of x^n-1}).
\end{enumerate}

In fact, for some codes, the Gray image is linear if and only if $\overline{f}$ and $\overline{g_p}$ are coprime.
The readers interested in the Gray images can characterize such cyclic codes and give the condition.
For this work, cyclotomic cosets may be necessary.
Moreover, it's also an interesting topic to find the condition for the Gray image of the case $p\mid n$ or a $\mathbb{Z}_{p^k}$-cyclic code to be linear. 

\section{Conflict of Interest}

The authors have no conflicts of interest to declare that are relevant to the content
of this article.

\section{Data Deposition Information}

Our data can be obtained from the authors upon reasonable request.


\end{document}